\documentclass[accepted]{uai2022} 

\usepackage[american]{babel}

\usepackage{natbib} 
\usepackage{mathtools} 
\usepackage{booktabs} 
\usepackage{tikz} 

\usepackage{url}            
\usepackage{amsfonts}       

\usepackage{amsmath}
\usepackage{amssymb}
\usepackage{mathrsfs}
\usepackage{amsthm}
\usepackage{thmtools}
\usepackage{thm-restate}
\usepackage{mathtools}
\usepackage{bbm}
\usepackage{centernot}

\usepackage{float}
\usepackage{caption}
\usepackage{subcaption}
\usepackage{color}
\usepackage{xcolor}
\usepackage{tabularx}
\usepackage[ruled,vlined]{algorithm2e}
\usepackage{tikz}
\usepackage{tkz-graph}
\usepackage{pgfplots}
\usepackage{pgffor}
\usetikzlibrary{arrows,arrows.meta,fit,automata,matrix,calc,patterns,decorations,decorations.markings,shapes,positioning}

\usepackage{xr}

\usepackage{enumitem}

\theoremstyle{plain}

\newtheorem*{theorem*}{Theorem}

\newtheorem*{proposition*}{Proposition}
\newtheorem{lemma}{Lemma}

\theoremstyle{definition}

\newtheorem*{remarkxx}{Remark}
\newtheorem{examplex}{Example}
\newtheorem*{examplexx}{Example}

\theoremstyle{plain}
\newtheorem{manualtheoreminner}{Theorem}
\newenvironment{manualtheorem}[1]{%
  \renewcommand\themanualtheoreminner{#1}%
  \manualtheoreminner
}{\endmanualtheoreminner}

\newenvironment{remark*}
{\pushQED{\qed}\remarkxx}
{\popQED\endremarkxx}

\newenvironment{example}
{\pushQED{\qed}\examplex}
{\popQED\endexamplex}
\newenvironment{example*}
{\pushQED{\qed}\examplexx}
{\popQED\endexamplexx}

\newcommand{\B}[1]{\boldsymbol{{#1}}}
\newcommand{\tuple}[1]{\langle {#1} \rangle}

\newcommand{\E}{\mathcal{E}}

\newcommand{\G}{\mathcal{G}}

\newcommand{\BG}{\mathcal{B}}
\newcommand{\V}{\mathcal{V}}

\newcommand{\adj}[2]{\mathrm{adj}_{#1}(#2)}

\newcommand{\indep}{\mathrel{\perp\mspace{-10mu}\perp}}
\newcommand{\dep}{\centernot{\indep}}

\title{Robustness of Model Predictions under Extension}

\author[1]{\href{mailto:<tineke_blom@hotmail.com>}{Tineke~Blom}{}}
\author[2]{\href{mailto:<j.m.mooij@uva.nl>}{Joris M.~Mooij}{}}
\affil[1]{%
    Informatics Institute\\
    University of Amsterdam\\
    The Netherlands
}
\affil[2]{%
    Korteweg-De Vries Institute for Mathematics\\
    University of Amsterdam\\
    The Netherlands
}
  
\begin{document}
\maketitle

\begin{abstract}%
Mathematical models of the real world are simplified representations of complex systems. A caveat to using mathematical models is that predicted causal effects and conditional independences may not be robust under model extensions, limiting applicability of such models. In this work, we consider conditions under which qualitative model predictions are preserved when two models are combined. Under mild assumptions, we show how to use the technique of \emph{causal ordering} to efficiently assess the robustness of qualitative model predictions. We also characterize a large class of model extensions that preserve qualitative model predictions. For dynamical systems at equilibrium, we demonstrate how novel insights help to select appropriate model extensions and to reason about the presence of feedback loops. We illustrate our ideas with a viral infection model with immune responses.
\end{abstract}


\section{Introduction}

A popular class of mathematical models that can represent uncertainty and causality are Structural Causal Models (SCMs) \citep{Pearl2009,Bongers++_AOS_21}.
However, there are several interesting systems for which the causal relations and Markov properties cannot be modelled by SCMs \citep{Blom2019}. The causal ordering algorithm, introduced by \citet{Simon1953}, can be used to deduce the qualitative predictions of mathematical models for these systems with regards to the causal relations between the variables in the system, and the probabilistic independence relations between the variables \citep{Blom2020}. In this paper, we take a closer look at what happens to these predictions when two systems are combined. Particularly, we give conditions under which properties of the whole system can be understood in terms of properties of its parts. We discuss how a holistic approach towards causal modelling may result in novel insights when we derive and test the predictions of systems for which new properties emerge from the combination of its parts.

In the first part of the paper, we focus on the practical issue of assessing whether qualitative model predictions are robust under model extensions. We revisit the observations of \citet{Boer2012} who demonstrated that qualitative predictions of a certain viral infection model change dramatically when the model is extended with extra equations describing simple immune responses. To assess the robustness of predicted causal relations or conditional independences under such an alteration of the model, it is useful to characterize a class of model extensions that lead to unaltered qualitative model predictions. In this work, we propose the technique of causal ordering \citep{Simon1953} as an efficient method to assess the robustness of qualitative causal predictions. Under mild conditions, this allows us to characterize a large class of model extensions that preserve qualitative causal predictions. We also consider the class of models that are obtained from the equilibrium equations of dynamical models where each variable is \emph{self-regulating}. For this class, we show that the predicted presence of causal relations and absence of conditional independences is robust when the model is extended with new equations.

Key aspects of the scientific method include generating a model or hypothesis that explains a phenomenon, deriving testable predictions from this model or hypothesis, and designing an experiment to test these predictions in the real world. The promise of causal discovery algorithms is that they are able to learn causal relations from a combination of background knowledge and data. The general idea of many constraint-based approaches (e.g. PC or FCI and variants thereof \citep{Spirtes2000, Zhang2008, Colombo2012}) is to exploit information about conditional independences in a probability distribution to construct an equivalence class of graphs that encode certain aspects of the probability distribution, and then draw conclusions about the causal relations from the graphs. There is a large amount of literature concerning particular algorithms for which the learned structure expresses causal relations under various combinations of assumptions (e.g.\ linearity, causal sufficiency, absence of feedback loops), see for example \citep{Richardson1999, Spirtes2000, Lacerda2008, Zhang2008, Colombo2012, Hyttinen2012, Forre2018, Strobl2018, Mooij2020}. In the last part of this paper, our main interest is in dynamical models with the property that directed graphs representing relations between variables by encoding the conditional independences of their equilibrium distribution should not be interpreted causally at all. For the case that a model for a subsystem is given, we present novel insights that enable us to reject model extensions based on conditional independences in equilibrium data of the subsystem. We demonstrate how this approach allows us to reason about the presence of variables that are not self-regulating and feedback mechanisms that involve unobserved variables from the equilibrium distribution of certain dynamical models. 

\subsection{Causal Ordering Graph}
\label{sec:introduction:causal ordering graph}

Here, 
we give a concise introduction to the technique of causal ordering,
introduced by \citet{Simon1952}.\footnote{Actually, 
we consider an equivalent algorithm for causal ordering that was shown to be
more computationally efficient by \citep{Nayak1995, Goncalves2016}. For more
details, see \citep{Blom2020}.} 
In short, the causal ordering algorithm takes a set of equations as input and
returns a \emph{causal ordering graph} that encodes the effects of
interventions and a \emph{Markov ordering graph} that implies conditional
independences between variables in the model \citep[Theorem 17, ][]{Blom2020}.  Compared with
the popular framework of structural causal models \citep{Pearl2009,Bongers++_AOS_21}, the
distinction between the causal ordering and Markov ordering graphs does not
provide new insights for acyclic models, but it results in non-trivial
conclusions for models with feedback, as suggested in the discussion in 
Section~\ref{sec:causal ordering for a viral infection model:causal interpretation}
and explained in detail in \citep{Blom2020}.

We consider models consisting of equations $F$ that contain endogenous
variables $V$, independent exogenous random variables $W$, and (constant, exogenous)
parameters $P$. The structure of equations and the endogenous variables that
appear in them can be represented by the \emph{associated bipartite graph}
$\BG=\tuple{V, F, E}$, where each endogenous variable is associated with a
distinct vertex in $V$, and each equation is associated with a distinct vertex
in $F$. There is an edge $(v-f)\in E$ if and only if variable $v\in V$ appears
in equation $f\in F$. The causal ordering algorithm constructs a \emph{directed
cluster graph} $\tuple{\V,\E}$, where $\V$ is a partition of vertices $V \cup F$ into
clusters and $\E$ is a set of directed edges from vertices in $V$ to clusters
in $\V$. Given a bipartite graph $\BG=\tuple{V,F,E}$ with a perfect matching\footnote{A 
perfect matching $M$ is a subset of edges in a bipartite graph so that every
vertex is adjacent to exactly one edge in $M$. If $(v-f) \in M$ we say that 
$v$ and $f$ are matched in $M$. Note that not every bipartite
graph has a perfect matching; the theory can be extended to bipartite
graphs that have no perfect matching by making use of maximal matchings instead
\citep{Blom2020}.}
$M$, the causal ordering algorithm proceeds with the following three
steps \citep{Nayak1995, Blom2020}:

\begin{enumerate}
  \item For $v\in V$, $f\in F$ orient edges $(v-f)$ as $(v\leftarrow f)$ when $(v-f)\in M$ and as $(v\rightarrow f)$ otherwise; this yields a directed graph $\G(\BG,M)$. \label{alg:step 1}
  \item Find all strongly connected components $S_1,S_2,\ldots, S_n$ of $\G(\BG,M)$. Let $\V$ be the set of clusters $S_i\cup M(S_i)$ for $i\in\{1,\ldots,n\}$, where $M(S_i)$ denotes the set of vertices that are matched to vertices in $S_i$ in matching $M$. \label{alg:step 2}
  \item Let $\mathrm{cl}(f)$ denote the cluster in $\V$ containing $f$. For each $(v\to f)$ in $\G(\BG,M)$ such that $v\notin \mathrm{cl}(f)$ add an edge $(v\to\mathrm{cl}(f))$ to $\E$.\label{alg:step 3}
\end{enumerate}

Independent exogenous random variables and parameters are then added
as singleton clusters with edges towards the clusters of the equations in which
they appear. It has been shown that the resulting directed cluster graph
$\mathrm{CO}(\BG)=\tuple{\V,\E}$, which we refer to as the \emph{causal
ordering graph}, is independent of the choice of perfect matching \citep[Theorem 4, ][]{Blom2020}. 
Example~\ref{ex:model equations as bipartite graph} shows how
the algorithm works and a graphical illustration of the algorithm for a more
elaborate cyclic model can be found in Section~\ref{sec:supplement:example} of the Supplementary Material.

\begin{example}
  \label{ex:model equations as bipartite graph}
  Let $V=\{v_1,v_2\}$, $W=\{w_1,w_2\}$, and $P=\{p_1,p_2\}$ be index sets. Consider model equations $f_1$ and $f_2$ with endogenous variables $(X_v)_{v\in V}$, exogenous random variables $(U_w)_{w\in W}$ and parameters $(C_p)_{p\in P}$:
  \begin{align}
  f_1:&\qquad C_{p_1} X_{v_1} - U_{w_1} = 0,\\
  f_2:&\qquad C_{p_2} X_{v_2} + X_{v_1} + U_{w_2} =0.
  \end{align}

  The bipartite graph $\BG=\tuple{V,F,E}$ in Figure~\ref{fig:example
  bipartite}, with $E=\{(v_1-f_1), (v_1-f_2), (v_2-f_2)\}$ is a compact
  representation of the model structure. This graph has a perfect matching
  $M=\{(v_1-f_1), (v_2-f_2)\}$. By orienting edges in $\BG$ according to the
  rules in step \ref{alg:step 1} of the causal ordering algorithm we obtain the
  directed graph $\tuple{V\cup F,E_{\mathrm{dir}}}$ with
  $E_{\mathrm{dir}}=\{(f_1\to v_1), (f_2\to v_2), (v_1\to f_2)\}$. The clusters
  $C_1=\{v_1,f_1\}$ and $C_2=\{v_2,f_2\}$ are added to $\V$ in
  step~\ref{alg:step 2} of the algorithm, and the edge $(v_1\to C_2)$ is added
  to $\E$ in step~\ref{alg:step 3}. Finally, we add the parameters $P$ and
  independent exogenous random variables $W$ as singleton clusters to $\V$, and
  the edges $(p_1\to C_1)$, $(w_1 \to C_1)$, $(p_2\to C_2)$, and $(w_2\to C_2)$
  to $\E$. The resulting causal ordering graph is given in
  Figure~\ref{fig:example causal ordering}.
\end{example}

Throughout this work, we will assume that models are \emph{uniquely solvable with
respect to the causal ordering graph}, which roughly means that for each
cluster, the equations in that cluster can be solved uniquely for the
endogenous variables in that cluster (see \citep[Definition 14, ][]{Blom2020}
for details). A \emph{perfect intervention on a cluster} that contains equation
vertices represents a model change where the equations in the targeted cluster
are replaced by equations that set the endogenous variables in that cluster
equal to constant values. A \emph{soft intervention} targets an equation, parameter,
or exogenous variable, but does not affect which variables appear in the
equations. We say that there is a \emph{directed path} from a vertex $x$ to a vertex
$y$ in a causal ordering graph $\tuple{\V,\E}$ if either
$\mathrm{cl}(x)=\mathrm{cl}(y)$ or there is a sequence of clusters
$C_1=\mathrm{cl}(x),C_2,\ldots, C_{k-1}, C_k=\mathrm{cl}(y)$ so that for all
$i\in\{1,\ldots,k-1\}$ there is a vertex $z_i\in C_i$ such that $(z_i\to
C_{i+1})\in\mathcal{E}$. It can be shown that (a) the presence of a directed
path from a cluster, equation, parameter, or exogenous variable that is
targeted by a soft intervention towards a certain variable in the causal
ordering graph implies that the intervention has a generic effect on that
variable, and (b) if no such path exists there is no causal effect of the
intervention on that variable \citep[Theorem 20, ][]{Blom2020}. 
For a perfect intervention that targets a cluster in the causal ordering graph, 
one can similarly read off its non-effects and generic effects from the causal 
ordering graph \citep[Theorem 23, ][]{Blom2020}.

\begin{figure}[ht]
\begin{subfigure}[b]{0.5\linewidth}
  \centering
  \begin{tikzpicture}[scale=0.75,every node/.style={transform shape}]  
  \GraphInit[vstyle=Normal]
  \Vertex[L=$v_1$]{v1}
  \EA[unit=1.25,L=$v_2$](v1){v2}
  \SO[unit=1, L=$f_1$](v1){f1}
  \SO[unit=1, L=$f_2$](v2){f2}
  \tikzset{EdgeStyle/.style = {-}}
  \Edge(v1)(f1)
  \Edge(v2)(f2)
  \Edge(v1)(f2)
  \end{tikzpicture}
  \caption{Bipartite graph.}
  \label{fig:example bipartite}
\end{subfigure}%
\begin{subfigure}[b]{0.5\linewidth}
  \centering
  \begin{tikzpicture}[scale=0.75,every node/.style={transform shape}]  
  \GraphInit[vstyle=Normal]
  
  \Vertex[L=$w_1$, style=dashed]{w1}
  \begin{scope}[VertexStyle/.append style = {minimum size = 4pt, 
    inner sep = 0pt,
    color=black}]
  \SO[unit=1., L=$p_1$, Lpos=270, LabelOut](w1){p1}
  \end{scope}
  \EA[unit=1.25, L=$v_1$](w1){v1}
  \EA[unit=1.25,L=$v_2$](v1){v2}
  \SO[unit=1, L=$f_1$](v1){f1}
  \SO[unit=1, L=$f_2$](v2){f2}
  \EA[unit=1.25, L=$w_2$, style=dashed](v2){w2}
  \begin{scope}[VertexStyle/.append style = {minimum size = 4pt, 
    inner sep = 0pt,
    color=black}]
  \SO[unit=1, L=$p_2$, Lpos=270, LabelOut](w2){p2}
  \end{scope}
  
  \node[draw=black, fit=(v1) (f1), inner sep=0.1cm ]{};
  \node[draw=black, fit=(v2) (f2), inner sep=0.1cm ]{};
  
  \draw[EdgeStyle, style={->}](v1) to (1.975,0);
  \draw[EdgeStyle, style={->}](w1) to (0.725,0);
  \draw[EdgeStyle, style={->}](p1) to (0.725,-1);
  \draw[EdgeStyle, style={->}](w2) to (3.025,0);
  \draw[EdgeStyle, style={->}](p2) to (3.025,-1);    
  \end{tikzpicture}
  \caption{Causal ordering graph.}
  \label{fig:example causal ordering}
\end{subfigure}

\medskip
\begin{subfigure}[b]{\linewidth}
  \centering
  \begin{tikzpicture}[scale=0.75,every node/.style={transform shape}]  
  \GraphInit[vstyle=Normal]
  
  \Vertex[L=$w_1$, style=dashed]{w1}
  \EA[unit=1.25, L=$v_1$](w1){v1}
  \EA[unit=1.25,L=$v_2$](v1){v2}
  \EA[unit=1.25, L=$w_2$, style=dashed](v2){w2}
  
  \draw[EdgeStyle, style={->}](v1) to (v2);
  \draw[EdgeStyle, style={->}](w1) to (v1);
  \draw[EdgeStyle, style={->}](w2) to (v2);
  \end{tikzpicture}
  \caption{Markov ordering graph.}
  \label{fig:example markov ordering}
\end{subfigure}
\caption{The bipartite graph in Figure~\ref{fig:example bipartite} is a compact representation of the model in Example~\ref{ex:model equations as bipartite graph}. The corresponding causal ordering graph and Markov ordering graph are given in Figures \ref{fig:example causal ordering} and \ref{fig:example markov ordering} respectively. Exogenous variables are denoted by dashed circles and parameters by black dots.}
\end{figure}
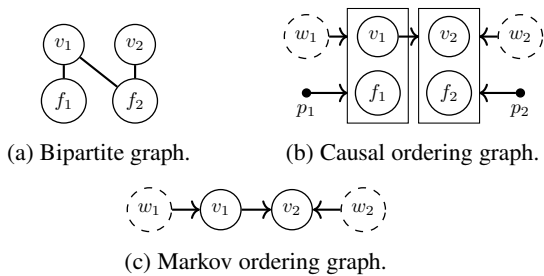

\subsection{Markov Ordering Graph}
\label{sec:introduction:markov ordering graph}

The causal ordering graph $\mathrm{CO}(\BG)=\tuple{\V,\E}$ of model equations
$F$ with endogenous variables $V$, exogenous random variables $W$, 
parameters $P$, and bipartite graph $\BG$ can be used to construct the
\emph{Markov ordering graph}, which is a DAG $\mathrm{MO}(\BG)=\tuple{V\cup
W,E}$, with $(x\to y)\in E$ if and only if $(x\to \mathrm{cl}(y))\in
\mathcal{E}$. The Markov ordering graph for the model equations in
Example~\ref{ex:model equations as bipartite graph} is given in
Figure~\ref{fig:example markov ordering}. It has been shown that, under the
assumption of unique solvability w.r.t.\ the causal ordering graph,
d-separations in the Markov ordering graph imply conditional independences
between the corresponding variables \citep{Blom2020}. Henceforth, we will
assume that the probability distribution of the solution $(X_v)_{v\in V}$ to a
set of model equations is faithful to the Markov ordering graph. In other
words, each conditional independence in the distribution implies a d-separation
in the Markov ordering graph. Under the assumption that data is generated from
such a model, some causal discovery algorithms, such as the PC algorithm
\citep{Spirtes2000}, aim to construct the Markov equivalence class of the
Markov ordering graph. In this work, we will specifically focus on feedback
models for which the Markov ordering graph of the equilibrium distribution, and
consequently the output of many causal discovery algorithms, does not have a
straightforward causal interpretation.

\section{Causal Ordering for a Viral Infection Model}
\label{sec:causal ordering for a viral infection model}

This work was inspired by a viral infection model discussed by \citet{Boer2012}, who
showed through explicit calculations that the predictions of the model are not
robust under addition of an immune response. This sheds doubt on the correct
interpretation of variables and parameters in the model. For many systems it is
intrinsically difficult to study their behavior in detail. The use of
simplified mathematical models that capture key characteristics aids in the
analysis of certain properties of the system. The hope is that the
explanations inferred from model equations are legitimate accounts of the true
underlying system \citep{Boer2012}. In reality, a modeler must take into
account that the outcome of these studies may be contingent on the specifics of
the model design. Here, we demonstrate how causal ordering can be used as a
scalable tool to assess the robustness of model predictions without requiring
explicit calculations.

\subsection{Viral Infection without Immune Response}
\label{sec:causal ordering for a viral infection model:without}

Let $U_\sigma$ be a production term for target cells, $d_T$ the death rate for target cells, $U_f$ the fraction of successful infections, and $U_\delta$ the death rate of productively infected cells. Define $\beta=\frac{bp}{c}$, where $b$ is the infection rate, $p$ the amount of virus produced per infected cell, and $c$ the clearance rate of viral particles. The following first-order differential equations describe how the amount of target cells $X_T(t)$ and the amount of infected cells $X_I(t)$ evolve over time \citep{Boer2012}:
\begin{align}
\label{eq:simple T}
\dot{X}_T(t) &= U_\sigma - d_T X_T(t) - \beta X_T(t) X_I(t), \\
\label{eq:simple I}
\dot{X}_I(t) &= (U_f\beta X_T(t) - U_\delta) X_I(t).
\end{align}
Suppose that we want to use this simple viral infection model to explain why the \emph{set-point viral load} (i.e.\ the total amount of virus circulating in the bloodstream) of chronically infected HIV-patients differs by several orders of magnitude, as \citet{Boer2012} does. To analyze this problem we look at the equilibrium equations that are implied by equations \eqref{eq:simple T} and \eqref{eq:simple I}:\footnote{Since we are only interested in strictly positive solutions we removed $X_I$ from the equilibrium equation $f_I:\,\,(U_f\beta X_T - U_\delta ) X_I = 0$ to obtain $f_I^+$.}
\begin{align}
f_T:\qquad& U_\sigma - d_T X_T - \beta X_T X_I = 0, \label{eq:eq T}\\
f^+_I:\qquad& U_f\beta X_T - U_\delta = 0. \label{eq:eq I}
\end{align}
Throughout the remainder of this work we will use this \emph{natural labelling} of equilibrium equations, where the equation derived from the derivative $\dot{X}_i(t)$ is labelled $f_i$. For first-order differential equations that are written in canonical form, $\dot{X}_i(t)=g_i(X(t))$, the natural labelling always exists.

Suppose that $U_\sigma$, $U_f$ and $U_\delta$ are independent exogenous random variables taking values in $\mathbb{R}_{>0}$ and $d_T$, $\beta$ are strictly positive parameters. The associated bipartite graph, causal ordering graph, and Markov ordering graph are given in Figure~\ref{fig:viral infection}. The causal ordering graph tells us that soft interventions targeting $U_\sigma$, $U_f$, $U_\delta$, $d_T$, or $\beta$ generically have an effect on the equilibrium distribution of the amount of infected cells $X_I$. From here on, we say that the causal ordering graph of a model predicts the \emph{generic presence} or \emph{absence} of causal effects. The Markov ordering graph shows that $v_T$ and $w_{\sigma}$ are d-separated. This implies that the amount of target cells $X_T$ should be independent of the production rate $U_{\sigma}$ when the system is at equilibrium. Henceforth, we will say that the Markov ordering graph predicts the \emph{generic presence} or \emph{absence} of conditional dependences.

\begin{figure}[t]
  \centering
  \subcaptionbox{Bipartite graph.\label{fig:viral infection:bipartite graph}}{%
    \centering
    \begin{tikzpicture}[scale=0.75,every node/.style={transform shape}]
    \GraphInit[vstyle=Normal]
    \Vertex[L=$v_T$]{vT} 
    \EA[unit=1,L=$v_I$](vT){vI}
    \SO[unit=1.1,L=$f_T$](vT){fT}
    \SO[unit=1.1,L=$f^+_I$](vI){fI}
    \tikzset{EdgeStyle/.style = {-}}
    \Edge(fI)(vT)
    \Edge(fT)(vI)
    \Edge(fT)(vT)
    \end{tikzpicture}\qquad\qquad
    }
    \subcaptionbox{Causal ordering graph.\label{fig:viral infection:causal ordering graph}}{%
    \centering
    \begin{tikzpicture}[scale=0.75,every node/.style={transform shape}]
    \GraphInit[vstyle=Normal]
    \Vertex[L=$v_T$]{vT} 
    \EA[unit=2.5,L=$v_I$](vT){vI}
    \SO[unit=1,L=$f^+_I$](vT){fI}
    \SO[unit=1,L=$f_T$](vI){fT}
    \node[draw=black, fit=(vT) (fI), inner sep=0.1cm ]{};
    \node[draw=black, fit=(vI) (fT), inner sep=0.1cm ]{};
    
    \Vertex[x=4,y=0.0, L=$w_{\sigma}$, style=dashed]{s}
    \Vertex[x=-1.5,y=0, L=$w_f$, style=dashed]{f}
    \Vertex[x=-1.5,y=-1, L=$w_\delta$, style=dashed]{d}
    
    \begin{scope}[VertexStyle/.append style = {minimum size = 4pt, 
      inner sep = 0pt,
      color=black}]
    \Vertex[x=1.25,y=-0.3, Lpos=270, LabelOut, L=$d_T$]{dt}
    \Vertex[x=1.25,y=-1.2, Lpos=270, LabelOut, L=$\beta$]{b}
    \end{scope}
    
    \draw[EdgeStyle, style={->}](s) to (3.12,0.0);
    \draw[EdgeStyle, style={->}](dt) to (1.88,-0.3);
    \draw[EdgeStyle, style={->}](b) to (1.88,-1.2);
    
    \draw[EdgeStyle, style={->}](vT) to (1.88,0.0);
    \draw[EdgeStyle, style={->}](b) to (0.67,-1.2);
    
    \draw[EdgeStyle, style={->}](f) to (-0.67,0.0);
    \draw[EdgeStyle, style={->}](d) to (-0.67,-1.0);
    \end{tikzpicture}
    }

    \medskip
    \subcaptionbox{Markov ordering graph.\label{fig:viral infection:markov ordering graph}}{%
    \centering
    \begin{tikzpicture}[scale=0.75,every node/.style={transform shape}]
    \GraphInit[vstyle=Normal]
    \Vertex[L=$v_T$]{vT}
    \EA[unit=1.25,L=$v_I$](vT){vI}
    \EA[unit=1.25,L=$w_{\sigma}$, style=dashed](vI){s}
    \SO[unit=1.25, L=$w_{\delta}$, style=dashed](vT){d}
    
    \Vertex[x=-1.25,y=0.0, L=$w_{f}$, style=dashed]{f}

    \tikzset{EdgeStyle/.style = {->}}
    \Edge(d)(vT)
    \Edge(vT)(vI)
    \Edge(s)(vI)
    \Edge(f)(vT)
    \end{tikzpicture}
    }
\caption{Graphical representations of the viral infection model in equations \eqref{eq:eq T} and \eqref{eq:eq I}. Vertices $v_i$ and $w_j$ correspond to variables $X_i$ and $U_j$, respectively. The causal ordering graph represents generic effects of interventions. The d-separations in Figure~\ref{fig:viral infection:markov ordering graph} imply conditional independences.}
\label{fig:viral infection}
\end{figure}
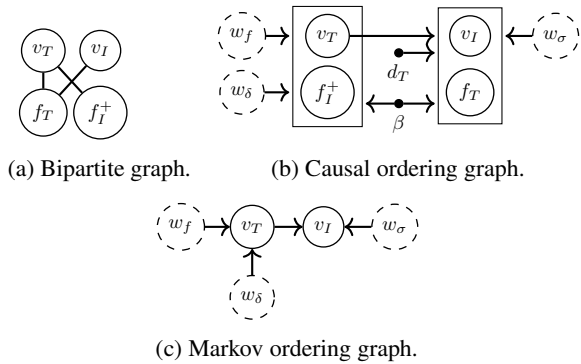

\subsection{Viral Infection with a Single Immune Response}
\label{sec:causal ordering for a viral infection model:single}

The viral infection model in equations \eqref{eq:simple T} and \eqref{eq:simple I} can be extended with a simple immune response $X_E(t)$ by adding the following dynamic and static equations:
\begin{align}
\label{eq:simple E}
\dot{X}_E(t) &= (U_a X_I(t) - d_E) X_E(t), \\
\label{eq:simple delta}
X_\delta(t) &= d_I + U_k X_E(t),
\end{align}
where $U_a$ is an activation rate, $d_E$ and $d_I$ are turnover rates and $U_k$ is a mass-action killing rate \citep{Boer2012}. Note that the exogenous random variable $U_\delta$ is now treated as an endogenous variable $X_{\delta}(t)$ instead. We derive the following equilibrium equations, using the natural labelling for equation \eqref{eq:simple E}:\footnote{Analogous to changing $f_I$ to $f_I^+$ for strictly positive solutions, we will look at $f_E^+$ instead of $f_E$.}
\begin{align}
\label{eq:eq E}
f_E^+: \qquad & U_a X_I - d_E = 0, \\
\label{eq:delta}
f_{\delta}: \qquad & X_\delta - d_I - U_k X_E = 0,
\end{align}
Henceforth, we will call the addition of equations $F_+$ to $F$ a \emph{model
extension}. Notice that when two sets of equations are combined, there may
exist variables that were exogenous in the submodel (i.e.\ the original model)
but that are endogenous within the whole model (i.e.\ the extended model).
Generally, equations $F_+$ may contain endogenous variables in $V$ and
exogenous variables in $W$ but they may also contain additional endogenous
variables $V_+$, additional exogenous variables $W_+$ and additional parameters $P_+$.
Parameters and exogenous random variables that appear in equations $F$ can appear as
endogenous variables in $V_+$ and in the extended model $F_{\mathrm{ext}}=F\cup
F_+$. In that case, these variables are no longer considered to be parameters
or exogenous variables within the extended model.

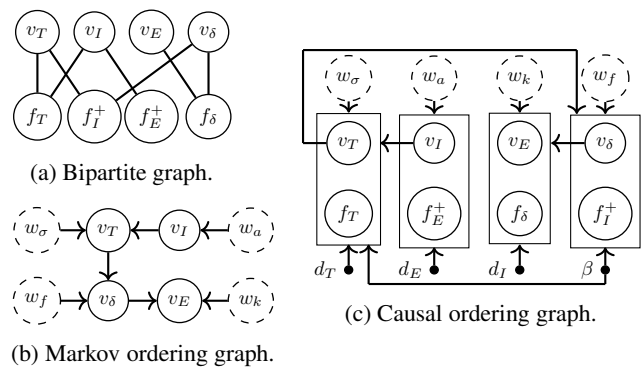
\begin{figure}[ht]
  \begin{minipage}{0.45\linewidth}
  \subcaptionbox{Bipartite graph.\label{fig:extended viral infection:bipartite graph}}{%
    \centering
    \begin{tikzpicture}[scale=0.75,every node/.style={transform shape}]
    \GraphInit[vstyle=Normal]
    \Vertex[L=$v_T$]{vT} 
    \EA[unit=1,L=$v_I$](vT){vI}
    \EA[unit=1,L=$v_E$](vI){vE}
    \EA[unit=1,L=$v_{\delta}$](vE){vd}
    \SO[unit=1.5,L=$f_T$](vT){fT}
    \SO[unit=1.5,L=$f_I^+$](vI){fI}
    \SO[unit=1.5,L=$f_E^+$](vE){fE}
    \SO[unit=1.5,L=$f_{\delta}$](vd){fd}
    \tikzset{EdgeStyle/.style = {-}}
    \Edge(fI)(vT)
    \Edge(fI)(vd)
    \Edge(fT)(vI)
    \Edge(fT)(vT)
    \Edge(fE)(vI)
    \Edge(fd)(vE)
    \Edge(fd)(vd)
    \end{tikzpicture}
    }

    \medskip
  \subcaptionbox{Markov ordering graph.\label{fig:extended viral infection:markov ordering graph}}{%
    \centering
    \begin{tikzpicture}[scale=0.75,every node/.style={transform shape}]
    \GraphInit[vstyle=Normal]
    \Vertex[L=$v_T$]{vT}
    \EA[unit=1.25,L=$v_I$](vT){vI}
    \SO[unit=1.25, L=$v_{\delta}$](vT){d}
    \SO[unit=1.25, L=$v_E$](vI){vE}
    
    \Vertex[x=-1.25, y=0.0 ,L=$w_{\sigma}$, style=dashed]{s}
    \Vertex[x=2.5, y=0.0 ,L=$w_{a}$, style=dashed]{a}
    \Vertex[x=2.5, y=-1.25 ,L=$w_{k}$, style=dashed]{k}
    \Vertex[x=-1.25, y=-1.25 ,L=$w_{f}$, style=dashed]{f}
    
    \tikzset{EdgeStyle/.style = {->}}
    \Edge(vT)(d)
    \Edge(vI)(vT)
    \Edge(d)(vE)
    \Edge(s)(vT)
    \Edge(a)(vI)
    \Edge(k)(vE)
    \Edge(f)(d)
    \end{tikzpicture}
    }
\end{minipage}
  \begin{minipage}{0.45\linewidth}
    \subcaptionbox{Causal ordering graph.\label{fig:extended viral infection:causal ordering graph}}{%
    \centering
    \begin{tikzpicture}[scale=0.75,every node/.style={transform shape}]   
    \GraphInit[vstyle=Normal]
    \Vertex[L=$v_T$]{vT} 
    \EA[unit=1.5,L=$v_I$](vT){vI}
    \EA[unit=1.5,L=$v_E$](vI){vE}
    \EA[unit=1.5,L=$v_{\delta}$](vE){vd}
    \SO[unit=1.25,L=$f_T$](vT){fT}
    \SO[unit=1.25,L=$f_E^+$](vI){fE}
    \SO[unit=1.25,L=$f_{\delta}$](vE){fd}
    \SO[unit=1.25,L=$f_I^+$](vd){fI}
    \node[draw=black, fit=(vT) (fT), inner sep=0.1cm ]{};
    \node[draw=black, fit=(vI) (fE), inner sep=0.1cm ]{};
    \node[draw=black, fit=(vE) (fd), inner sep=0.1cm ]{};
    \node[draw=black, fit=(vd) (fI), inner sep=0.1cm ]{};
    
    \Vertex[x=0.0,y=1.15, style=dashed, L=$w_{\sigma}$]{s}
    \Vertex[x=1.5,y=1.15, style=dashed, L=$w_{a}$]{a}
    \Vertex[x=3.0,y=1.15, style=dashed, L=$w_{k}$]{k}
    \Vertex[x=4.5,y=1.15, style=dashed, L=$w_{f}$]{f}
    
    \begin{scope}[VertexStyle/.append style = {minimum size = 4pt, 
      inner sep = 0pt,
      color=black}]
    \Vertex[x=0.0,y=-2.25, Lpos=180, LabelOut, L=$d_T$]{dt}
    \Vertex[x=1.5,y=-2.25, Lpos=180, LabelOut, L=$d_E$]{de}
    \Vertex[x=3.0,y=-2.25, LabelOut, Lpos=180, L=$d_I$]{di}
    \Vertex[x=4.5,y=-2.25, LabelOut, Lpos=180, L=$\beta$]{b}
    \end{scope}
    
    \draw[EdgeStyle, style={->}](s) to (0.0,0.54);
    \draw[EdgeStyle, style={->}](a) to (1.5,0.515);
    \draw[EdgeStyle, style={->}](k) to (3.0,0.54);
    \draw[EdgeStyle, style={->}](f) to (4.5,0.515);
    
    \draw[EdgeStyle, style={->}](dt) to (0.0,-1.8);
    \draw[EdgeStyle, style={->}](de) to (1.5,-1.85);
    \draw[EdgeStyle, style={->}](di) to (3.0,-1.8);
    \draw[EdgeStyle, style={->}](b) to (4.5,-1.85);
    
    \draw[EdgeStyle, style={->}](vI) to (0.55,0);
    \draw[EdgeStyle, style={->}](vd) to (3.55,0);

    \draw[EdgeStyle, style={-}](b) to (4.5,-2.5);
    \draw[EdgeStyle, style={-}](4.5,-2.5) to (0.35,-2.5);
    \draw[EdgeStyle, style={->}](0.35,-2.5) to (0.35,-1.8);
    
    \draw[EdgeStyle, style={-}](vT) to (-0.8,0.0);
    \draw[EdgeStyle, style={-}](-0.8,0.0) to (-0.8,1.65);
    \draw[EdgeStyle, style={-}](-0.8,1.65) to (4.0, 1.65);
    \draw[EdgeStyle, style={->}](4.0, 1.65) to (4.0, 0.515);    
    \end{tikzpicture} 
    }
\end{minipage}
  \caption{Graphical representations of the viral infection model with a single immune response. The presence or absence of causal relations and d-connections implied by the graphs in Figure~\ref{fig:viral infection} are not preserved if a single immune response is added.}
  \label{fig:extended viral infection}
\end{figure}

Suppose that $U_a$ and $U_k$ are independent exogenous random variables taking
values in $\mathbb{R}_{>0}$ and $d_E, d_I$ are parameters taking value in
$\mathbb{R}_{>0}$. The bipartite graph, causal ordering graph, and Markov
ordering graph associated with equations \eqref{eq:eq T}, \eqref{eq:eq I},
\eqref{eq:eq E}, and \eqref{eq:delta} (with $X_\delta$ replacing $U_\delta$)
are given in Figure~\ref{fig:extended viral infection}. 

The causal ordering graph in Figure~\ref{fig:extended viral infection:causal ordering graph}
predicts a causal effect of $U_{\sigma}$ and $d_T$ on $X_T$ but not on $X_I$. 
By comparing with the predictions of the causal ordering graph in
Figure~\ref{fig:viral infection:causal ordering graph} (where we saw that
soft interventions targeting $U_\sigma$ and $d_T$ generically do have an effect on $X_I$), we find that effects of
interventions targeting $U_{\sigma}$ and $d_T$ are not robust under the model
extension. 

The Markov ordering graph of the extended model
in Figure~\ref{fig:extended viral infection:markov ordering graph} shows that
$w_{\sigma}$ is d-connected to $v_T$, and hence $U_\sigma$ and $X_T$ will be
dependent at equilibrium for most parameter choices. On the other hand, in
the Markov ordering graph of the viral infection model without immune response
(Figure~\ref{fig:viral infection:markov ordering graph}), 
$w_\sigma$ and $v_T$ are d-separated, and hence, $U_\sigma$ and $X_T$ will be independent 
at equilibrium for any parameter choice according to the viral infection model without immune response.
Therefore, the independence between $U_{\sigma}$ and $X_T$ is not robust under the model extension.

The systematic graphical procedure followed here easily leads to the same causal conclusions as \citet{Boer2012} obtained by explicitly solving the equilibrium equations. In addition, it leads to predictions regarding the conditional (in)dependences in the equilibrium distribution.

\subsection{Viral infection with Multiple Immune Responses}
\label{sec:causal ordering for a viral infection model:multiple}

The following static and dynamical equations describe multiple immune responses:
\begin{align}
  \begin{split}
\dot{X}_{E_i}(t) &= \frac{p_E X_{E_i}(t) U_{a_i} X_I(t)}{h+ X_{E_i}(t) + U_{a_i} X_I(t)} - d_E X_{E_i}(t), \\
    & \quad\qquad\qquad\qquad\qquad\qquad i=1,2,\ldots,n 
  \end{split}\\
X_{\delta}(t) &= d_I + U_k \sum_{i=1}^{n} U_{a_i} X_{E_i}(t),
\end{align}
where there are $n$ immune responses, $U_{a_i}$ is the avidity of immune
response $i$, $p_E$ is the maximum division rate, and $h$ is a saturation
constant \citep{Boer2012}. For $n=2$ we can derive equilibrium equations
$f_{E_1}$, $f_{E_2}$, and $f_{\delta}$ using the natural labelling as we did
for the equilibrium equations in the previous section. Together with the
equilibrium equations \eqref{eq:eq T} and \eqref{eq:eq I} (with $X_\delta$
replacing $U_\delta$) for the viral infection model this is another extended
model. The bipartite graph of this extended model is given in
Figure~\ref{fig:multiple immune responses:bipartite graph}, while the causal
ordering graph can be found in Figure~\ref{fig:multiple immune responses:causal
ordering graph}. By comparing the directed paths in this causal ordering graph
with that of the original viral infection model (i.e.\ the model without an
immune response) in Figure~\ref{fig:viral infection:causal ordering graph}, it
can be seen that the predicted presence of causal relations is preserved under
extension of the model with multiple immune responses, while the predicted
absence of causal relations is not. Similarly, by comparing d-separations in
the Markov ordering graphs in Figure~\ref{fig:viral infection:markov ordering
graph} with those in Figure~\ref{fig:multiple immune responses:markov ordering
graph}, we find that predicted conditional dependences are preserved under the
extensions, while the predicted conditional independences are not.

\begin{figure}[h]
  \centering
  \begin{subfigure}[b]{\linewidth}
    \centering
    \begin{tikzpicture}[scale=0.7,every node/.style={transform shape}]
    \GraphInit[vstyle=Normal]
    \Vertex[L=$v_T$]{vT} 
    \EA[unit=1.25,L=$v_I$](vT){vI}
    \EA[unit=1.25,L=$v_{E_1}$](vI){vE1}
    \EA[unit=1.25,L=$v_{E_2}$](vE1){vE2}
    \EA[unit=1.25,L=$v_{\delta}$](vE2){vd}
    \SO[unit=1.75,L=$f_T$](vT){fT}
    \SO[unit=1.75,L=$f_I^+$](vI){fI}
    \SO[unit=1.75,L=$f_{E_1}$](vE1){fE1}
    \SO[unit=1.75,L=$f_{E_2}$](vE2){fE2}
    \SO[unit=1.75,L=$f_{\delta}$](vd){fd}
    \tikzset{EdgeStyle/.style = {-}}
    \Edge(fI)(vT)
    \Edge(fI)(vd)
    \Edge(fT)(vI)
    \Edge(fT)(vT)
    \Edge(fE1)(vI)
    \Edge(fE1)(vE1)
    \Edge(fE2)(vI)
    \Edge(fE2)(vE2)
    \Edge(fd)(vE1)
    \Edge(fd)(vE2)
    \Edge(fd)(vd)
    \end{tikzpicture}
    \caption{Bipartite graph.}
    \label{fig:multiple immune responses:bipartite graph}
  \end{subfigure}

  \medskip
  \begin{subfigure}[b]{\linewidth}
    \centering
    \begin{tikzpicture}[scale=0.75,every node/.style={transform shape}]  
    \GraphInit[vstyle=Normal]
    \Vertex[L=$v_T$]{vT} 
    \EA[unit=1,L=$v_I$](vT){vI}
    \EA[unit=1,L=$v_{E_1}$](vI){vE1}
    \EA[unit=1,L=$v_{E_2}$](vE1){vE2}
    \EA[unit=1,L=$v_{\delta}$](vE2){vd}
    \SO[unit=1,L=$f_I^+$](vT){fI}
    \SO[unit=1,L=$f_T$](vI){fT}
    \SO[unit=1,L=$f_{E_1}$](vE1){fE1}
    \SO[unit=1,L=$f_{E_2}$](vE2){fE2}
    \SO[unit=1,L=$f_{\delta}$](vd){fd}
    \node[draw=black, fit=(vT) (fI) (vI) (fT) (vE1) (fE1) (vd) (fd), inner sep=0.1cm ]{};
    
    \Vertex[x=0,y=1.3, L=$w_{\sigma}$, style=dashed]{s}
    \Vertex[x=1,y=1.3, L=$w_{k}$, style=dashed]{k}
    \Vertex[x=2,y=1.3, L=$w_{a_1}$, style=dashed]{a1}
    \Vertex[x=3,y=1.3, L=$w_{a_2}$, style=dashed]{a2}
    \Vertex[x=4,y=1.3, L=$w_f$, style=dashed]{f}
    
    \begin{scope}[VertexStyle/.append style = {minimum size = 4pt, 
      inner sep = 0pt,
      color=black}]
    \Vertex[x=-1,y=0.25, Lpos=180, LabelOut, L=$d_T$]{dt}
    \Vertex[x=-1,y=-0.5, Lpos=180, LabelOut, L=$\beta$]{b}
    \Vertex[x=-1,y=-1.25, Lpos=180,LabelOut, L=$h$]{h}
    
    \Vertex[x=5,y=0.25,Lpos=0, LabelOut, L=$p_E$]{pe}
    \Vertex[x=5,y=-0.5,Lpos=0,LabelOut, L=$d_E$]{de}
    \Vertex[x=5,y=-1.25,Lpos=0, LabelOut, L=$d_I$]{di}
    \end{scope}
    \draw[EdgeStyle, style={->}](s) to (0,0.6);
    \draw[EdgeStyle, style={->}](k) to (1,0.6);
    \draw[EdgeStyle, style={->}](a1) to (2,0.6);
    \draw[EdgeStyle, style={->}](a2) to (3,0.6);
    \draw[EdgeStyle, style={->}](f) to (4,0.6);
    
    \draw[EdgeStyle, style={->}](dt) to (-0.6,0.25);
    \draw[EdgeStyle, style={->}](b) to (-0.6,-0.5);
    \draw[EdgeStyle, style={->}](h) to (-0.6,-1.25);
    
    \draw[EdgeStyle, style={->}](pe) to (4.54,0.25);
    \draw[EdgeStyle, style={->}](de) to (4.54,-0.5);
    \draw[EdgeStyle, style={->}](di) to (4.54, -1.25);
    
    \end{tikzpicture}
   \caption{Causal ordering graph.}
   \label{fig:multiple immune responses:causal ordering graph}
  \end{subfigure}

  \medskip
  \begin{subfigure}[b]{\linewidth}
    \centering
    \begin{tikzpicture}[scale=0.75,every node/.style={transform shape}] 
    \GraphInit[vstyle=Normal]
    \Vertex[L=$v_T$]{vT} 
    \EA[unit=1.5,L=$v_I$](vT){vI}
    \EA[unit=1.5,L=$v_{E_1}$](vI){vE1}
    \EA[unit=1.5,L=$v_{E_2}$](vE1){vE2}
    \EA[unit=1.5,L=$v_{\delta}$](vE2){vd}
    
    \Vertex[x=0.0,y=1.5, L=$w_{\sigma}$, style=dashed]{s}
    \Vertex[x=1.5,y=-1.5, L=$w_{k}$, style=dashed]{k}
    \Vertex[x=3.0,y=1.5, L=$w_{a_1}$, style=dashed]{a1}
    \Vertex[x=4.5,y=-1.5, L=$w_{a_2}$, style=dashed]{a2}
    \Vertex[x=6.0,y=1.5, L=$w_f$, style=dashed]{f}
    
    \draw[EdgeStyle, style={->}](s) to (vT);
    \draw[EdgeStyle, style={->}](s) to (vI);
    \draw[EdgeStyle, style={->}](s) to (vE1);
    \draw[EdgeStyle, style={->}](s) to (vE2);
    \draw[EdgeStyle, style={->,bend left=6}](s) to (vd);
    
    \draw[EdgeStyle, style={->}](k) to (vT);
    \draw[EdgeStyle, style={->}](k) to (vI);
    \draw[EdgeStyle, style={->}](k) to (vE1);
    \draw[EdgeStyle, style={->}](k) to (vE2);
    \draw[EdgeStyle, style={->}](k) to (vd);
    
    \draw[EdgeStyle, style={->, bend right=8}](a1) to (vT);
    \draw[EdgeStyle, style={->}](a1) to (vI);
    \draw[EdgeStyle, style={->}](a1) to (vE1);
    \draw[EdgeStyle, style={->}](a1) to (vE2);
    \draw[EdgeStyle, style={->, bend left=8}](a1) to (vd);
    
    \draw[EdgeStyle, style={->}](a2) to (vT);
    \draw[EdgeStyle, style={->}](a2) to (vI);
    \draw[EdgeStyle, style={->}](a2) to (vE1);
    \draw[EdgeStyle, style={->}](a2) to (vE2);
    \draw[EdgeStyle, style={->}](a2) to (vd);
    
    \draw[EdgeStyle, style={->, bend right=6}](f) to (vT);
    \draw[EdgeStyle, style={->}](f) to (vI);
    \draw[EdgeStyle, style={->}](f) to (vE1);
    \draw[EdgeStyle, style={->}](f) to (vE2);
    \draw[EdgeStyle, style={->}](f) to (vd);
    \end{tikzpicture}
    \caption{Markov ordering graph.}
    \label{fig:multiple immune responses:markov ordering graph}
  \end{subfigure}
  \caption{Graphical representations of the viral infection model with multiple immune responses. The presence of causal relations and d-connections in Figure~\ref{fig:viral infection} is preserved.}
  \label{fig:multiple immune responses}
\end{figure}
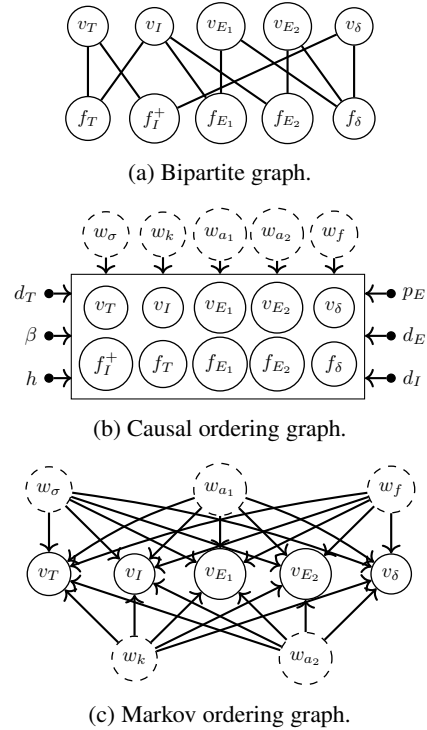

\subsection{Markov Ordering Graphs and Causal Interpretations}
\label{sec:causal ordering for a viral infection model:causal interpretation}

In \citep{Blom2020}, it was shown that the Markov ordering graph may not have a straightforward causal interpretation. Here, we illustrate for the viral infection models that
the Markov ordering graphs neither have a straightforward causal interpretation at equilibrium in terms of soft interventions targeting parameters or equations \emph{nor} in terms of perfect interventions on variables in the dynamical model. To see this, consider the Markov ordering graph in Figure~\ref{fig:extended viral infection:markov ordering graph} for the viral infection with a single immune response. The edge $(v_I\to v_T)$ cannot correspond to the effect of a soft intervention targeting $f_I^+$, because the causal ordering graph in Figure~\ref{fig:extended viral infection:causal ordering graph} shows that there is no such effect. Clearly, directed paths in the Markov ordering graph do not necessarily represent the effects of soft interventions. A natural way to model a perfect intervention targeting a variable in the Markov ordering graph is to replace the (differential) equation of that variable with an equation setting that variable equal to a certain value in the underlying dynamical model \citep{Mooij2013}. By explicitly solving equilibrium equations it is easy to check that replacing $f_E^+$ with an equation setting $X_E$ equal to a constant generically changes the equilibrium distributions of all four variables $X_E, X_\delta, X_T, X_I$. Since there are no directed paths from $v_{E}$ to any of $v_\delta, v_T, v_I$ in the Markov ordering graph in Figure~\ref{fig:extended viral infection:markov ordering graph}, the effects of this perfect intervention would not have been predicted by the Markov ordering graph, if it were interpreted causally. Therefore, contrary to the causal ordering graph, the Markov ordering graph does not have a causal interpretation in terms of soft or perfect interventions on the underlying dynamical model.

\section{Robust Causal Predictions under Model Extensions}
\label{sec:robust causal predictions}

One way to gauge the robustness of model predictions is to check to what extent they depend on the model design. The example of a viral infection with different immune responses in the previous section indicates that qualitative causal predictions entailed by the causal ordering graph of a mathematical model may strongly depend on the particulars of the model. Both the implied presence or absence of causal relations at equilibrium and the implied presence or absence of conditional independences at equilibrium may change under certain model extensions. Under what conditions are these qualitative predictions preserved under model extensions? In this section, we characterize a large class of model extensions under which qualitative equilibrium predictions are preserved.

Theorem~\ref{thm:preserved causal predictions} gives a sufficient condition on
model extensions under which the predicted generic presence of causal relations
and predicted generic presence of conditional dependences at equilibrium is
preserved. The proof is given in Section~\ref{sec:supplement:proofs} of the Supplementary Material.

\begin{restatable}{theorem}{preservedcausalpredictions}
  \label{thm:preserved causal predictions}
  Consider model equations $F$ containing endogenous variables $V$ with bipartite graph $\BG$. Suppose $F$ is extended with equations $F_+$ containing endogenous variables in $V\cup V_+$, where $V_+$ contains endogenous variables that are added by the model extension (which may include parameters or exogenous variables that appear in $F$ and become endogenous in the extended model). Let $\BG_{\mathrm{ext}}$ be the bipartite graph associated with $F_{\mathrm{ext}}=F\cup F_+$ and $V_{\mathrm{ext}}=V\cup V_+$, and $\BG_+$ the bipartite graph associated with the extension $F_+$ and $V_+$, where variables in $V$ appearing in $F^+$ are treated as exogenous variables (i.e.\ they are not added as vertices in $\BG_+$). If $\BG$ and $\BG_+$ both have a perfect matching then:
  \begin{enumerate}
    \item $\BG_{\mathrm{ext}}$ has a perfect matching, \label{con 1: thm1}
    \item ancestral relations in $\mathrm{CO}(\BG)$ are also present in $\mathrm{CO}(\BG_{\mathrm{ext}})$, \label{con 2: thm1}
    \item d-connections in $\mathrm{MO}(\BG)$ are also present in $\mathrm{MO}(\BG_{\mathrm{ext}})$. \label{con 3: thm1}
  \end{enumerate}
\end{restatable}

This result characterizes a large set of extensions under which the implied
causal effects and conditional dependences of a model are preserved. Consider
again the equilibrium behavior of the viral infection models in
Section~\ref{sec:causal ordering for a viral infection model}. We already
showed explicitly that the extension of the viral infection model with multiple
immune responses preserved the predicted presence of causal relations and
conditional dependences, but with the help of Theorem~\ref{thm:preserved causal
predictions} we only would have needed to check whether the bipartite graph in
Figure~\ref{fig:model extension:multiple responses} has a perfect matching to
arrive at the same conclusion. The bipartite graph for the extension with a
single immune response in Figure~\ref{fig:model extension:single response} does
not have a perfect matching and hence the conditions of
Theorem~\ref{thm:preserved causal predictions} do not hold. Recall that this
model extension did not preserve the predicted presence of causal relations.

\begin{figure}[t]
  \begin{subfigure}[b]{0.4\linewidth}
    \centering
    \begin{tikzpicture}[scale=0.7,every node/.style={transform shape}] 
    \GraphInit[vstyle=Normal]
    \Vertex[L=$v_T$,style=dotted]{vT} 
    \EA[unit=1.25,L=$v_I$,style=dotted](vT){vI}
    \EA[unit=1.25,L=$v_E$](vI){vE} 
    \EA[unit=1.25,L=$v_{\delta}$](vE){vd}
    \SO[unit=1.75,L=$f_T$,style=dotted](vT){fT}
    \SO[unit=1.75,L=$f_I^+$,style=dotted](vI){fI}
    \SO[unit=1.75,L=$f_E^+$](vE){fE}
    \SO[unit=1.75,L=$f_{\delta}$](vd){fd}
    \tikzset{EdgeStyle/.style = {-}}
    \Edge[style=dotted](fI)(vT)
    \Edge[style=dotted](fI)(vd)
    \Edge[style=dotted](fT)(vI)
    \Edge[style=dotted](fT)(vT)
    \Edge[style=dotted](fE)(vI)
    \Edge(fd)(vE)
    \Edge(fd)(vd)
    \end{tikzpicture}  
    \caption{Single response.}
    \label{fig:model extension:single response}
  \end{subfigure}\hfill
  \begin{subfigure}[b]{0.55\linewidth}
    \centering
    \begin{tikzpicture}[scale=0.7,every node/.style={transform shape}]
    \GraphInit[vstyle=Normal]
    \Vertex[L=$v_T$,style=dotted]{vT} 
    \EA[unit=1.25,L=$v_I$,style=dotted](vT){vI}
    \EA[unit=1.25,L=$v_{E_1}$](vI){vE1}
    \EA[unit=1.25,L=$v_{E_2}$](vE1){vE2}
    \EA[unit=1.25,L=$v_{\delta}$](vE2){vd}
    \SO[unit=1.75,L=$f_T$,style=dotted](vT){fT}
    \SO[unit=1.75,L=$f_I^+$,style=dotted](vI){fI}
    \SO[unit=1.75,L=$f_{E_1}$](vE1){fE1}
    \SO[unit=1.75,L=$f_{E_2}$](vE2){fE2}
    \SO[unit=1.75,L=$f_{\delta}$](vd){fd}
    \tikzset{EdgeStyle/.style = {-}}
    \Edge[style=dotted](fI)(vT)
    \Edge[style=dotted](fI)(vd)
    \Edge[style=dotted](fT)(vI)
    \Edge[style=dotted](fT)(vT)
    \Edge[style=dotted](fE1)(vI)
    \Edge(fE1)(vE1)
    \Edge[style=dotted](fE2)(vI)
    \Edge(fE2)(vE2)
    \Edge(fd)(vE1)
    \Edge(fd)(vE2)
    \Edge(fd)(vd)
    \end{tikzpicture}
    \caption{Multiple responses.}
    \label{fig:model extension:multiple responses}
  \end{subfigure}
  \caption{The solid lines indicate the bipartite graphs $\BG_+$ associated with the single immune response extension (\ref{fig:model extension:single response}) and the multiple immune response extension (\ref{fig:model extension:multiple responses}). Combined with the dotted lines, one obtains the bipartite graphs $\BG_{\mathrm{ext}}$ for the complete models.}
  \label{fig:model extension}
\end{figure}
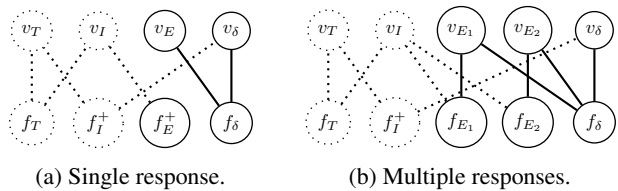

The theorem below gives a stronger condition under which (conditional) independence relations and the absence of causal relations that are implied by a model are also predicted by the extended model. The proof is provided in the supplement.

\begin{restatable}{theorem}{preservedcis}
  \label{thm:preserveration of conditional independence relations}
  Let $F$, $F_+$, $F_{\mathrm{ext}}$, $V$, $V_+$, $V_{\mathrm{ext}}$, $\BG$, $\BG_+$, and $\BG_{\mathrm{ext}}$ be as in Theorem~\ref{thm:preserved causal predictions}. If $\BG$ and $\BG_+$ both have perfect matchings and no vertex in $V_+$ is adjacent to a vertex in $F$ in $\BG_{\mathrm{ext}}$ then:\footnote{A vertex in $V_+$ is considered adjacent to $F$ in $\BG_{\mathrm{ext}}$ if it corresponds with one of the exogenous random variables or parameters in $F$ that become endogenous in the model extension.}
  \begin{enumerate}
    \item ancestral relations absent in $\mathrm{CO}(\BG)$ are also absent in $\mathrm{CO}(\BG_{\mathrm{ext}})$, \label{res:2}
    \item d-connections absent in $\mathrm{MO}(\BG)$ are also absent in $\mathrm{MO}(\BG_{\mathrm{ext}})$. \label{res:3}
  \end{enumerate}
\end{restatable}

Together with Theorem~\ref{thm:preserved causal predictions}, this result characterizes a large class of model extensions under which all
qualitative model predictions are preserved. Consider again the equilibrium
models for the viral infection in Section~\ref{sec:causal ordering for a viral
infection model}. The bipartite graph for the extension with a single immune
response in Figure~\ref{fig:model extension:single response} does not have a perfect matching. In the bipartite graph
associated with the viral infection model with multiple immune responses 
in Figure~\ref{fig:multiple immune responses:bipartite graph}, the
additional endogenous variable $v_{\delta}$ is adjacent to $f_I$. Neither of
the model extensions satisfies the conditions of
Theorem~\ref{thm:preserveration of conditional independence relations}. We
already demonstrated that neither of the model extensions preserves all
qualitative model predictions. An example of a model extension that does
satisfy the conditions in Theorem~\ref{thm:preserved causal predictions} and
\ref{thm:preserveration of conditional independence relations} is an acyclic
structural causal model that is extended with another acyclic structural causal
model such that the additional variables are non-ancestors of the original
ones. Together, Theorem~\ref{thm:preserved causal predictions} and
\ref{thm:preserveration of conditional independence relations}, can be used to
understand when the causal and Markov properties of a composite system can be 
understood by studying the corresponding properties of its components.

\section{Selection of Model Extensions}

So far, we have considered methods to assess the robustness of qualitative model predictions. In this section we will show how this idea results in novel opportunities regarding causal discovery. In particular, if we assume that the systems that we observe are part of a larger partially observed system, then we can use the methods in this paper to reason about causal mechanisms of unobserved variables. Consider, for example, the viral infection model for which we have demonstrated that extensions with different immune responses imply different (conditional) independences between variables in the original model. The Markov ordering graphs in Figures \ref{fig:viral infection:markov ordering graph}, \ref{fig:extended viral infection:markov ordering graph}, and \ref{fig:multiple immune responses:markov ordering graph} imply the following (in)dependences:

\medskip
\centerline{\begin{tabular}{ll}
  Viral infection model & (In)dependences \\
  \hline
  without immune response & $U_{\sigma} \indep X_T$, $U_{\sigma}\dep X_I$ \\
  with single immune response & $U_{\sigma} \dep X_T$, $U_{\sigma}\indep X_I$ \\
  with multiple immune responses & $U_{\sigma} \dep X_T$, $U_{\sigma}\dep X_I$ \\
\end{tabular}}

Given a model for variables $X_T$ and $X_I$ only, we can reject certain model
extensions based on the observed (conditional) independences for variables 
$X_T$, $X_I$, and $U_{\sigma}$ in data sampled from the distribution of the
combined system---provided that we assume faithfulness and the correctness of 
the model of the original subsystem.
Using this holistic modelling approach, we can reason about properties of
an unknown model extension without observing the new mechanisms or variables.
In the remainder of this section, we further discuss how this idea can be
applied to equilibrium data of dynamical systems.

\subsection{Reasoning about Self-regulating Variables}

We say that a variable in a set of first-order differential equations in
canonical form is \emph{self-regulating} if it can be solved uniquely from the
equilibrium equation that is constructed from its derivative. For example,
in the system of first-order ODEs (\ref{eq:simple T}--\ref{eq:simple I}),
$X_T$ is self-regulating if $\beta X_I + d_T \ne 0$, whereas $X_I$ is not
self-regulating.

For models in which every variable is self-regulating there exists a perfect matching where
each variable $v_i$ is matched to its associated equilibrium equation $f_i$
according to the natural labelling, for more details see
Lemma~\ref{lemma:selfregulating perfect matching} in the
supplement. Interestingly, the Markov ordering graph for the
equilibrium equations of such a model always has a causal interpretation. By
construction of the causal ordering graph from the bipartite graph and the
perfect matching provided by the natural labelling, we know that a vertex $v_i$
always appears in a cluster with $f_i$ in the causal ordering graph. The
presence or absence of directed paths in the Markov ordering graph can then
easily be associated with the presence or absence of directed paths in the
causal ordering graph. Consequently, the Markov ordering graph can be
interpreted in terms of both soft interventions targeting equations and perfect
interventions that set variables equal to a constant by replacement of the
associated dynamical and equilibrium equations.\footnote{Dynamical systems
with only self-regulating variables were also considered in \citep{Mooij2013},
where it was shown that their equilibria can be modelled as structural causal
models without self-cycles.}

For models in which every variable is self-regulating, it follows immediately 
from Theorem~\ref{thm:preserved
causal predictions} that the presence of ancestral relations and d-connections
is robust under dynamical model extensions in which each variable is
self-regulating, as is stated more formally in
Corollary~\ref{cor:selfregulating} below.

\begin{restatable}{corollary}{selfregulating}
\label{cor:selfregulating}
Consider a first-order dynamical model in canonical form for endogenous variables $V$ and an extension consisting of canonical first-order differential equations for additional endogenous variables $V_+$. Let $F$ and $F_{\mathrm{ext}}= F\cup F_+$ be the equilibrium equations of the original and extended model respectively. If all variables in $V\cup V_+$ are self-regulating, then statements \ref{con 2: thm1} and \ref{con 3: thm1} of Theorem~\ref{thm:preserved causal predictions} hold.
\end{restatable}

Corollary~\ref{cor:selfregulating} characterizes a class of models under which
certain qualitative predictions for the equilibrium distribution are robust, but the
result can also be interpreted from a different angle. Suppose that we have
equilibrium data that is generated by an extended dynamical model with
equilibrium equations $F_{\mathrm{ext}}$, but we only have a \emph{partial}
model consisting of equations in $F$ for a subset $V\subseteq
V_{\mathrm{ext}}=V\cup V_+$ of variables that appear in $F_{\mathrm{ext}}=F\cup F_+$. 
If we would find conditional independences between variables in $V$ that do not
correspond to d-separations in the Markov ordering graph of the partial model,
this does not necessarily mean that the model equations are wrong. It could
also be the case, for example, that we are wrong to assume that the system can
be studied in a reductionist manner and that the model should be extended.
Furthermore, under the assumption that data is generated from the equilibrium
distribution of a dynamical model, Corollary~\ref{cor:selfregulating} tells us
that conditional independences in the data that are not predicted by the
equations of a partial model imply the presence of variables that are not
self-regulating, if we assume faithfulness. This shows that, given a model for
a subsystem, we can reason about the properties of unobserved and unknown
variables in the whole system. 
We will showcase an example for this type of reasoning in Section~\ref{sec:signaling_cascade_model}.

\subsection{Reasoning about Feedback Loops}

We say that an extension of a dynamical model \emph{introduces a new feedback
loop with the original dynamical model} when there is feedback in the extended
dynamical model that involves variables in both the original model and the
model extension. To make this definition more precise, consider the set
$E_{\mathrm{nat}}$ of edges $(v_i-f_i)$ that are associated with the natural
labelling of the equilibrium equations of the extended dynamical model. The
feedback loops in the dynamical model coincide with cycles in the directed
graph $\G(\BG_{\mathrm{nat}}, M_{\mathrm{nat}})$ that is obtained by applying
step~\ref{alg:step 1} of the causal ordering algorithm to the bipartite graph
$\BG_{\mathrm{nat}}=\tuple{V_{\mathrm{ext}},F_{\mathrm{ext}},E_{\mathrm{ext}}\cup
E_{\mathrm{nat}}}$ using the perfect matching $M_{\mathrm{nat}} =
E_{\mathrm{nat}}$.\footnote{Note that a feedback loop in the dynamical model
does not imply a feedback loop in the equilibrium equations as well. For
example, there is feedback in the dynamical equations \eqref{eq:simple T},
\eqref{eq:simple I}, but there is no feedback in the causal ordering graph of
the equilibrium equations in Figure~\ref{fig:viral infection:causal ordering
graph} nor in the directed graph that is constructed in step~\ref{alg:step 1}
of the causal ordering algorithm.} The following theorem can be used to
reason about the presence of partially unobserved feedback loops given a model
and observations for a subsystem.

\begin{restatable}{theorem}{feedbackloop}
\label{thm:feedback}
Consider a first-order dynamical model in canonical form for endogenous variables $V$ and an extension consisting of canonical first-order differential equations for additional endogenous variables $V_+$. Let $F$ and $F_{\mathrm{ext}}= F\cup F_+$ be the equilibrium equations of the original and extended model respectively. Let $\BG=\tuple{V,F,E}$ be the bipartite graph associated with $F$ and $\BG_{\mathrm{ext}}=\tuple{V_{\mathrm{ext}}, F_{\mathrm{ext}}, E_{\mathrm{ext}}}$ the bipartite graph associated with $F_{\mathrm{ext}}$. Assume that $\BG$ and $\BG_{\mathrm{ext}}$ both have perfect matchings. If the model extension does not introduce a new feedback loop with the original dynamical model, then d-connections in $\mathrm{MO}(\BG)$ are also present in $\mathrm{MO}(\BG_{\mathrm{ext}})$.
\end{restatable}

Theorem~\ref{thm:feedback} characterizes a class of model extensions under
which certain qualitative model predictions are robust, but it also shows how we can
reason about the existence of unobserved feedback loops. To be more precise, it
shows that, given a submodel for a subsystem, the presence of conditional
independences that are not predicted by the submodel imply the existence of an
unobserved feedback loop, if we assume faithfulness. If, for example, we assume
that the viral infection model without an immune response is a submodel of the
system that is described by the strictly positive equilibrium solutions of the
viral infection model with a single immune response, then we would observe an
independence between $U_\sigma$ and $X_I$ that is not predicted by the model
equations of the submodel. Theorem~\ref{thm:feedback} would then imply
that there is an unobserved feedback loop. Indeed, it can be seen from
equations \eqref{eq:simple T}, \eqref{eq:simple I}, \eqref{eq:simple E},
\eqref{eq:simple delta} that there is an unobserved feedback loop from $X_I(t)$
to $X_E(t)$ to $X_\delta(t)$ and back to $X_I(t)$, while the Markov ordering
graphs in Figures \ref{fig:viral infection:markov ordering graph} and
\ref{fig:extended viral infection:markov ordering graph} imply that $U_\sigma$
and $X_I$ are dependent in the original model and independent in the extended
model. We consider the use of existing structure learning algorithms for the
detection of feedback loops in models with variables that are not
self-regulating from a combination of background knowledge and observational
equilibrium data to be an interesting topic for future work.


\subsection{Example: Signaling Cascade Model}\label{sec:signaling_cascade_model}

We will illustrate the ideas about detecting non-self-regulating variables and feedback loops by means of an example of a mathematical model for a dynamical system consisting of a signaling cascade of phosphorylated proteins.
The model is a simplified version of that of \citep{Shin2009}, where we omitted the feedback mechanism through RAF Kinase Inhibitor Protein (RKIP) \citep{BlomMooij_2101.11885}. 
The details of this model can be found in Section~\ref{sec:supplement:signaling_cascade_model} of the Supplementary Material.

We denote the concentrations of active (phosphorylated) RAS, RAF, MEK, and ERK proteins, respectively, by
$X_{s}(t)$, $X_{r}(t)$, $X_{m}(t)$, and $X_{e}(t)$, and denote by $I(t)$ an external stimulus or perturbation.
The dynamics is modeled by differential equations
\eqref{eq:mapk s},
\eqref{eq:mapk r},
\eqref{eq:mapk m}, and
\eqref{eq:mapk e approx}
in Section~\ref{sec:supplement:signaling_cascade_model} of the Supplementary Material. 
The full model consists of a signaling pathway that goes from $I(t)$ to $X_s(t)$ to $X_r(t)$ to $X_m(t)$ to $X_e(t)$, with negative feedback from $X_e(t)$ on $X_s(t)$.
At equilibrium, we assume the exogenous input signal $I(t) = I$ to have a constant (possibly random) value, and
let $f_s$, $f_r$, $f_m$, and $f_e$ represent the equilibrium equations \eqref{eq:mapk s eq}, \eqref{eq:mapk r eq}, \eqref{eq:mapk m eq}, and \eqref{eq:mapk e eq approx} respectively.

Suppose now that the system is only partially modelled by treating the ERK concentration $X_e(t)$ as a latent exogenous variable in the submodel for the RAS, RAF and MEK concentrations ($X_s(t)$, $X_r(t)$, and $X_m(t)$, respectively) defined by equations \eqref{eq:mapk s}, \eqref{eq:mapk r}, and \eqref{eq:mapk m}.
The complete model (including the differential equation \eqref{eq:mapk e approx} that models the dynamics of ERK) can then be seen as a model extension of the submodel for RAS, RAF and MEK.
Application of the causal ordering technique to the submodel (with $V = \{v_s, v_r, v_m\}$ and $F = \{f_s, f_r, f_m\}$) results in the Markov ordering graph in Figure~\ref{fig:protein pathway:submodel:markov ordering graph}. 
Assuming faithfulness, the d-connections in this graph indicate that the equilibrium distributions for $X_s$, $X_r$ and $X_m$ all depend on the input signal $I$.
Let us assume that we have observed data that is generated from the full model. 
The Markov ordering graph for the extended model (with $V^+ = \{v_e\}$, $F^+ = \{f_e\}$) is displayed in Figure~\ref{fig:protein pathway:extended model:markov ordering graph}, and implies that the equilibrium distributions for $X_s$, $X_r$ and $X_m$ are \emph{independent} of the input signal $I$. 
Thus, what appears to be a faithfulness violation from the submodel perspective is explained by the Markov property of the extended model. 
In this case, the holistic modeling approach that allows for feedback through additional unobserved endogenous variables is needed, while the more common reductionistic assumption of treating all unobserved causes as exogenous to the observed variables will fail.

According to Corollary~\ref{cor:selfregulating} and Theorem~\ref{thm:feedback}, the discrepancy between the observed and predicted conditional independences implies the presence of a non-selfregulating variable and an unobserved dynamical feedback loop (provided that we assume faithfulness). 
This is in agreement with the fact that the dynamic variable $X_e(t)$ is not self-regulating and that there is a feedback loop in the extended dynamical model. 
Remarkably, we can infer the presence of feedback in this way without explicitly modelling or even observing $X_e(t)$.

\begin{figure}[ht]
  \begin{subfigure}[b]{0.45\linewidth}
		\centering
		\begin{tikzpicture}[scale=0.75,every node/.style={transform shape}]
		\GraphInit[vstyle=Normal]
		\SetGraphUnit{1}
		\Vertex[L=$v_s$, x=1.1, y=0] {vS}
		\Vertex[L=$v_r$, x=2.2, y=0] {vR}
		\Vertex[L=$v_m$, x=3.3, y=0] {vM}
		\Vertex[L=$w_{s}$, x=1.1, y=-1.1, style={dashed}] {wS}
		\Vertex[L=$w_{r}$, x=2.2, y=-1.1, style={dashed}] {wR}
 		\Vertex[L=$w_{m}$, x=3.3, y=-1.1, style={dashed}] {wM}
		
		\begin{scope}[VertexStyle/.append style = {minimum size = 4pt, 
			inner sep = 0pt,
			color=black}]
		\Vertex[x=0.0, y=0.0, L=$I$, Lpos=180, LabelOut]{I}
		\end{scope}
    \draw[EdgeStyle, style={->}](I) to (vS);
    \draw[EdgeStyle, style={->}](wS) to (vS);
    \draw[EdgeStyle, style={->}](wR) to (vR);
    \draw[EdgeStyle, style={->}](wM) to (vM);
    \draw[EdgeStyle, style={->}](vR) to (vM);
    \draw[EdgeStyle, style={->}](vS) to (vR);
		\end{tikzpicture}
		\caption{$\mathrm{MO}(\BG)$.}
		\label{fig:protein pathway:submodel:markov ordering graph}
	\end{subfigure}
  \begin{subfigure}[b]{0.45\linewidth}
		\centering
		\begin{tikzpicture}[scale=0.75,every node/.style={transform shape}]
		\GraphInit[vstyle=Normal]
		\SetGraphUnit{1}
		\Vertex[L=$v_e$, x=0.0, y=0] {vE}
		\Vertex[L=$v_s$, x=1.1, y=0] {vS}
		\Vertex[L=$v_r$, x=2.2, y=0] {vR}
		\Vertex[L=$v_m$, x=3.3, y=0] {vM}
		\Vertex[L=$w_{s}$, x=0, y=-1.1, style={dashed}] {wS}
		\Vertex[L=$w_{r}$, x=1.1, y=-1.1, style={dashed}] {wR}
		\Vertex[L=$w_{m}$, x=2.2, y=-1.1, style={dashed}] {wM}
		\Vertex[L=$w_{e}$, x=3.3, y=-1.1, style={dashed}] {wE}
		
		\begin{scope}[VertexStyle/.append style = {minimum size = 4pt, 
			inner sep = 0pt,
			color=black}]
		\Vertex[x=-1.0, y=0.0, L=$I$, Lpos=180, LabelOut]{I}
		\end{scope}
		\draw[EdgeStyle, style={->}](I) to (vE);
		\draw[EdgeStyle, style={->}](wS) to (vE);
		\draw[EdgeStyle, style={->}](wR) to (vS);
		\draw[EdgeStyle, style={->}](wM) to (vR);
		\draw[EdgeStyle, style={->}](wE) to (vM);
		\draw[EdgeStyle, style={->}](vM) to (vR);
		\draw[EdgeStyle, style={->}](vR) to (vS);
		\draw[EdgeStyle, style={->}](vS) to (vE);
		\end{tikzpicture}
    \caption{$\mathrm{MO}(\BG_{\mathrm{ext}})$.}
		\label{fig:protein pathway:extended model:markov ordering graph}
	\end{subfigure}
  \caption{Markov ordering graphs for the partial model (left) and the full model (right) of the RAS-RAF-MEK-ERK protein signaling cascade model.}
\end{figure}
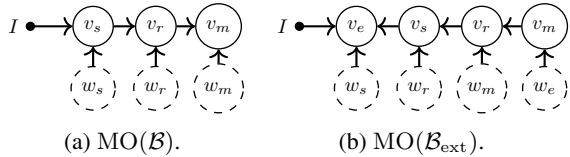

\section{Discussion and conclusion}
\label{sec:discussion}

In this work we revisited several models of viral infections and immune
responses. In our treatment of these models we closely followed the approach in
\citet{Boer2012} and therefore we only considered strictly positive solutions.
If we would have modelled all solutions then, for example, we would have
considered the equilibrium equation $f_I:(U_f\beta X_T - U_\delta)X_I = 0$
instead of $f_I^+$ in equation~\eqref{eq:eq I}. In that case, we would have
obtained the causal ordering graph in Figure~\ref{fig:all solutions} instead of
that in Figure~\ref{fig:viral infection:causal ordering graph}. Clearly, the
model predictions of the causal ordering graph for the positive solutions in
Figure~\ref{fig:viral infection:causal ordering graph} are more informative.
The choice of only modelling strictly positive solutions depends on the
application.

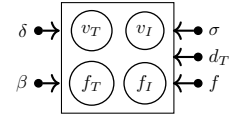
\begin{figure}[t]
	\centering
	\begin{tikzpicture}[scale=0.7,every node/.style={transform shape}] 
	\GraphInit[vstyle=Normal]
	\Vertex[L=$v_T$]{vT} 
	\EA[unit=1,L=$v_I$](vT){vI}
	\SO[unit=1,L=$f_T$](vT){fT}
	\SO[unit=1,L=$f_I$](vI){fI}
	\node[draw=black, fit=(vT) (fT) (vI) (fI), inner sep=0.1cm ]{};
	\begin{scope}[VertexStyle/.append style = {minimum size = 4pt, 
		inner sep = 0pt,
		color=black}]
	\Vertex[x=2,y=0.0, LabelOut, L=$\sigma$]{s}
	\Vertex[x=2,y=-0.5, LabelOut, L=$d_T$]{dt}
	\Vertex[x=2,y=-1.0, LabelOut, L=$f$]{f}
	\Vertex[x=-1.0,y=0.0, LabelOut, Lpos=180, L=$\delta$]{d}
	\Vertex[x=-1.0,y=-1.0, LabelOut, Lpos=180, L=$\beta$]{b}
	\end{scope}
	\draw[EdgeStyle, style={->}](s) to (1.525,0.0);
	\draw[EdgeStyle, style={->}](dt) to (1.525,-0.5);
	\draw[EdgeStyle, style={->}](f) to (1.525,-1.0);
	\draw[EdgeStyle, style={->}](d) to (-0.6,0.0);
	\draw[EdgeStyle, style={->}](b) to (-0.6,-1.0);
	\end{tikzpicture}
	\caption{Causal ordering graph for positive and non-positive solutions of the viral infection model.}
	\label{fig:all solutions}
\end{figure}

In many application domains mathematical models are used to predict the equilibrium behavior of complex systems. An important issue is that (causal and Markov) predictions may strongly depend on the specifics of the model design. We revisited an example of a viral infection model \citep{Boer2012}, in which implied causal relations and conditional independences change dramatically when equations, describing immune reactions, are added. Analysis of this behavior through explicit calculations is neither insightful nor scalable. We showed how the technique of causal ordering can be used to efficiently analyze the robustness of implied causal effects and conditional independences under certain solvability assumptions. Using key insights provided by this approach we characterized large classes of model extensions under which predicted causal relations and conditional independences are robust. We hope that the results presented in this paper provide a step towards bringing the world of causal modeling and reasoning closer to practical applications.

Our results for the characterization of the robustness of model extensions can
also be used to reason about the properties of models that are the combination
of two submodels. This way, we can study systems whose causal and Markov
properties can be understood in a reductionistic manner by considering the
properties of its parts. When the properties of the whole model differ from
those of its parts, a holistic modelling approach is required. For models
of the equilibrium distribution of dynamical systems, we proved that extensions
of dynamical models where each variable is self-regulating preserve the
predicted presence of causal effects and d-connections in the original model.
Based on those insights, we proposed a novel approach to model selection, where
information about conditional independences can be used in combination with
model equations to reason about possible model extensions or the presence of
feedback mechanisms. For dynamical models with feedback, the output of
structure learning algorithms does not always have a causal interpretation in
terms of soft or perfect interventions for the equilibrium distribution. We
have shown that in dynamical systems where each variable is self-regulating the
identifiable directed edges in the learned graph do express causal relations
between variables.

\begin{acknowledgements} 
We thank Johannes Textor for introducing us to De Boer's paper and for interesting discussions about the viral infection model. 
We are also grateful to an anonymous reviewer for pointing out some flaws in an earlier version of this manuscript. This work was supported by the ERC under the European Union's Horizon 2020 research and innovation programme (grant agreement 639466).
\end{acknowledgements}

\bibliography{blom_190}

\clearpage
\onecolumn
\appendix

\section{Supplementary Material}

\renewcommand{\theequation}{\arabic{equation}*}
\renewcommand{\thefigure}{\arabic{figure}*}
\renewcommand{\theHequation}{\arabic{equation}*}
\renewcommand{\theHfigure}{\arabic{figure}*}

This supplementary material contains the material that did not fit into the main paper because of space constraints.
A graphical illustration of the causal ordering algorithm applied to the equations of a cyclic model is provided in the first section. The second section contains more details on the signaling cascade model. The third section contains the proofs of the results in the main paper.

\subsection{Causal Ordering Algorithm Applied to a Cyclic Model}
\label{sec:supplement:example}

In this supplementary section we demonstrate how the causal ordering algorithm works on a set of equations for a cyclic model. The algorithm is also presented graphically. Consider the following equations for endogenous variables $\B{X}$ and exogenous random variables $\B{U}$:
\begin{align}
f_1:\qquad & g_1(X_{v_1}, U_{w_1}) = 0, \label{eq:f1}\\
f_2:\qquad & g_2(X_{v_2}, X_{v_1}, X_{v_4}, U_{w_2}) = 0,\\
f_3:\qquad & g_3(X_{v_3}, X_{v_2}, U_{w_3}) = 0,\\
f_4:\qquad & g_4(X_{v_4}, X_{v_3}, U_{w_4}) = 0,\\
f_5:\qquad & g_5(X_{v_5}, X_{v_4}, U_{w_5}) = 0. \label{eq:f5}
\end{align}

The associated bipartite graph in Figure~\ref{fig:supplement bipartite graph} consists of variable vertices $V=\{v_1,\ldots,v_5\}$ and equation vertices $F=\{f_1,\ldots,f_5\}$. There is an edge between a variable vertex and an equation vertex whenever that variable appears in the equation. The associated bipartite graph has exactly two perfect matchings:
\begin{align*}
M_1 = \{(v_1-f_1), (v_2-f_2), (v_3-f_3), (v_4-f_4), (v_5-f_5)\},\\
M_2 = \{(v_1-f_1), (v_2-f_3), (v_3-f_4), (v_4-f_2), (v_5-f_5)\}.
\end{align*}
Application of the first step of the causal ordering algorithm results either
in the directed graph in Figure~\ref{fig:supplement directed graph 1} or that
in Figure~\ref{fig:supplement directed graph 2}, depending on the choice of the
perfect matching. The segmentation of vertices into strongly connected
components, which takes place in the second step of the algorithm, results in
the clusters $\{v_1\}$, $\{f_1\}$, $\{v_2,v_3,v_4,f_2,f_3,f_4\}$, $\{v_5\}$,
and $\{f_5\}$. To construct the clusters of the causal ordering graph we add
$S_i\cup M(S_i)$ to a cluster set $\V$ for each $S_i$ in the segmentation. The
segmentation of vertices into strongly connected components is displayed in
Figures \ref{fig:supplement strongly connected components 1} and
\ref{fig:supplement strongly connected components 2}. Notice that the
segmentation in Figure~\ref{fig:supplement strongly connected components 1} is
the same as that in Figure~\ref{fig:supplement strongly connected components 2}.
It is known that the segmentation into strongly connected components is
unique (i.e.\ it does not depend on the choice of the perfect matching) \citep{Pothen1990,Blom2020}. 
The cluster set $\V$ for the causal ordering graph in Figure~\ref{fig:supplement causal ordering graph} is constructed by merging clusters in the segmented graph whenever two clusters contain vertices that are matched and by adding exogenous variables as singleton clusters. The edge set $\E$ for the causal ordering graph is obtained by adding edges $(v\to C)$ from an endogenous vertex $v$ to a cluster $C$, whenever $v\notin C$ and there is an edge from $v$ to $f\in C$ in the directed graph. We also add edges from exogenous vertices to clusters that contain equations in which the corresponding exogenous random variables appear.

\begin{figure}[ht]
\centering
\begin{subfigure}[b]{0.3\linewidth}
  \centering
  \begin{tikzpicture}[scale=0.75,every node/.style={transform shape}]
  \GraphInit[vstyle=Normal]
  \SetGraphUnit{1}
  \Vertex[L=$v_1$, x=0.0, y=0.0] {v1}
  \Vertex[L=$v_2$, x=1.0, y=0.0] {v2}
  \Vertex[L=$v_3$, x=2.0, y=0.0] {v3}
  \Vertex[L=$v_4$, x=3.0, y=0.0] {v4}
  \Vertex[L=$v_5$, x=4.0, y=0.0] {v5}

  \Vertex[L=$f_1$, x=0.0, y=-1.4] {f1}
  \Vertex[L=$f_2$, x=1.0, y=-1.4] {f2}
  \Vertex[L=$f_3$, x=2.0, y=-1.4] {f3}
  \Vertex[L=$f_4$, x=3.0, y=-1.4] {f4}
  \Vertex[L=$f_5$, x=4.0, y=-1.4] {f5}

  \draw[EdgeStyle, style={-}](v1) to (f1);
  \draw[EdgeStyle, style={-}](v2) to (f2);
  \draw[EdgeStyle, style={-}](v3) to (f3);
  \draw[EdgeStyle, style={-}](v4) to (f4);
  \draw[EdgeStyle, style={-}](v5) to (f5);
  \draw[EdgeStyle, style={-}](v1) to (f2);
  \draw[EdgeStyle, style={-}](v2) to (f3);
  \draw[EdgeStyle, style={-}](v3) to (f4);
  \draw[EdgeStyle, style={-}](v4) to (f5);
  \draw[EdgeStyle, style={-}](v4) to (f2);
  \end{tikzpicture}
  \caption{Bipartite graph.}
  \label{fig:supplement bipartite graph}
\end{subfigure}\hfill
\begin{subfigure}[b]{0.3\linewidth}
  \centering
  \begin{tikzpicture}[scale=0.75,every node/.style={transform shape}]
  \GraphInit[vstyle=Normal]
  \SetGraphUnit{1}
  \Vertex[L=$v_1$, x=0.0, y=0.0] {v1}
  \Vertex[L=$v_2$, x=1.0, y=0.0] {v2}
  \Vertex[L=$v_3$, x=2.0, y=0.0] {v3}
  \Vertex[L=$v_4$, x=3.0, y=0.0] {v4}
  \Vertex[L=$v_5$, x=4.0, y=0.0] {v5}

  \Vertex[L=$f_1$, x=0.0, y=-1.4] {f1}
  \Vertex[L=$f_2$, x=1.0, y=-1.4] {f2}
  \Vertex[L=$f_3$, x=2.0, y=-1.4] {f3}
  \Vertex[L=$f_4$, x=3.0, y=-1.4] {f4}
  \Vertex[L=$f_5$, x=4.0, y=-1.4] {f5}

  \draw[EdgeStyle, style={stealth-, blue, ultra thick}](v1) to (f1);
  \draw[EdgeStyle, style={stealth-, blue, ultra thick}](v2) to (f2);
  \draw[EdgeStyle, style={stealth-, blue, ultra thick}](v3) to (f3);
  \draw[EdgeStyle, style={stealth-, blue, ultra thick}](v4) to (f4);
  \draw[EdgeStyle, style={stealth-, blue, ultra thick}](v5) to (f5);
  \draw[EdgeStyle, style={-stealth, gray}](v1) to (f2);
  \draw[EdgeStyle, style={-stealth, gray}](v2) to (f3);
  \draw[EdgeStyle, style={-stealth, gray}](v3) to (f4);
  \draw[EdgeStyle, style={-stealth, gray}](v4) to (f5);
  \draw[EdgeStyle, style={-stealth, gray}](v4) to (f2);
  \end{tikzpicture}
  \caption{Directed graph ($M_1$).}
  \label{fig:supplement directed graph 1}
\end{subfigure}\hfill
\begin{subfigure}[b]{0.3\linewidth}
  \centering
  \begin{tikzpicture}[scale=0.75,every node/.style={transform shape}]
  \GraphInit[vstyle=Normal]
  \SetGraphUnit{1}
  \Vertex[L=$v_1$, x=0.0, y=0.0] {v1}
  \Vertex[L=$v_2$, x=1.0, y=0.0] {v2}
  \Vertex[L=$v_3$, x=2.0, y=0.0] {v3}
  \Vertex[L=$v_4$, x=3.0, y=0.0] {v4}
  \Vertex[L=$v_5$, x=4.0, y=0.0] {v5}

  \Vertex[L=$f_1$, x=0.0, y=-1.4] {f1}
  \Vertex[L=$f_2$, x=1.0, y=-1.4] {f2}
  \Vertex[L=$f_3$, x=2.0, y=-1.4] {f3}
  \Vertex[L=$f_4$, x=3.0, y=-1.4] {f4}
  \Vertex[L=$f_5$, x=4.0, y=-1.4] {f5}

  \draw[EdgeStyle, style={stealth-, orange, ultra thick}](v1) to (f1);
  \draw[EdgeStyle, style={-stealth}](v2) to (f2);
  \draw[EdgeStyle, style={-stealth}](v3) to (f3);
  \draw[EdgeStyle, style={-stealth}](v4) to (f4);
  \draw[EdgeStyle, style={stealth-, orange, ultra thick}](v5) to (f5);
  \draw[EdgeStyle, style={-stealth}](v1) to (f2);
  \draw[EdgeStyle, style={stealth-, orange, ultra thick}](v2) to (f3);
  \draw[EdgeStyle, style={stealth-, orange, ultra thick}](v3) to (f4);
  \draw[EdgeStyle, style={-stealth}](v4) to (f5);
  \draw[EdgeStyle, style={stealth-, orange, ultra thick}](v4) to (f2);
  \end{tikzpicture}
  \caption{Directed graph ($M_2$).}
  \label{fig:supplement directed graph 2}
\end{subfigure}%

\medskip
\begin{subfigure}[b]{0.3\linewidth}
  \centering
  \begin{tikzpicture}[scale=0.65,every node/.style={transform shape}]
  \GraphInit[vstyle=Normal]
  \SetGraphUnit{1}
  \Vertex[L=$v_1$, x=0.0, y=0.0] {v1}
  \Vertex[L=$v_2$, x=1.2, y=0.0] {v2}
  \Vertex[L=$v_3$, x=2.2, y=0.0] {v3}
  \Vertex[L=$v_4$, x=3.2, y=0.0] {v4}
  \Vertex[L=$v_5$, x=4.4, y=0.0] {v5}

  \Vertex[L=$f_1$, x=0.0, y=-1.4] {f1}
  \Vertex[L=$f_2$, x=1.2, y=-1.4] {f2}
  \Vertex[L=$f_3$, x=2.2, y=-1.4] {f3}
  \Vertex[L=$f_4$, x=3.2, y=-1.4] {f4}
  \Vertex[L=$f_5$, x=4.4, y=-1.4] {f5}

  \node[draw=black, fit=(v1), inner sep=0.1cm ]{};
  \node[draw=black, fit=(f1), inner sep=0.1cm ]{};
  \node[draw=black, fit=(v2) (f2) (v3) (f3) (v4) (f4), inner sep=0.1cm ]{};
  \node[draw=black, fit=(v5), inner sep=0.1cm ]{};
  \node[draw=black, fit=(f5), inner sep=0.1cm ]{};

  \draw[EdgeStyle, style={stealth-, blue, ultra thick}](v1) to (f1);
  \draw[EdgeStyle, style={stealth-, blue, ultra thick}](v2) to (f2);
  \draw[EdgeStyle, style={stealth-, blue, ultra thick}](v3) to (f3);
  \draw[EdgeStyle, style={stealth-, blue, ultra thick}](v4) to (f4);
  \draw[EdgeStyle, style={stealth-, blue, ultra thick}](v5) to (f5);
  \draw[EdgeStyle, style={-stealth, gray}](v1) to (f2);
  \draw[EdgeStyle, style={-stealth, gray}](v2) to (f3);
  \draw[EdgeStyle, style={-stealth, gray}](v3) to (f4);
  \draw[EdgeStyle, style={-stealth, gray}](v4) to (f5);
  \draw[EdgeStyle, style={-stealth, gray}](v4) to (f2);
  \end{tikzpicture}
  \caption{Segmentation ($M_1$)}
  \label{fig:supplement strongly connected components 1}
\end{subfigure}\hfill
\begin{subfigure}[b]{0.3\linewidth}
  \centering
  \begin{tikzpicture}[scale=0.65,every node/.style={transform shape}]
  \GraphInit[vstyle=Normal]
  \SetGraphUnit{1}
  \Vertex[L=$v_1$, x=0.0, y=0.0] {v1}
  \Vertex[L=$v_2$, x=1.2, y=0.0] {v2}
  \Vertex[L=$v_3$, x=2.2, y=0.0] {v3}
  \Vertex[L=$v_4$, x=3.2, y=0.0] {v4}
  \Vertex[L=$v_5$, x=4.4, y=0.0] {v5}

  \Vertex[L=$f_1$, x=0.0, y=-1.4] {f1}
  \Vertex[L=$f_2$, x=1.2, y=-1.4] {f2}
  \Vertex[L=$f_3$, x=2.2, y=-1.4] {f3}
  \Vertex[L=$f_4$, x=3.2, y=-1.4] {f4}
  \Vertex[L=$f_5$, x=4.4, y=-1.4] {f5}

  \node[draw=black, fit=(v1), inner sep=0.1cm ]{};
  \node[draw=black, fit=(f1), inner sep=0.1cm ]{};
  \node[draw=black, fit=(v2) (f2) (v3) (f3) (v4) (f4), inner sep=0.1cm ]{};
  \node[draw=black, fit=(v5), inner sep=0.1cm ]{};
  \node[draw=black, fit=(f5), inner sep=0.1cm ]{};

  \draw[EdgeStyle, style={stealth-, orange, ultra thick}](v1) to (f1);
  \draw[EdgeStyle, style={-stealth}](v2) to (f2);
  \draw[EdgeStyle, style={-stealth}](v3) to (f3);
  \draw[EdgeStyle, style={-stealth}](v4) to (f4);
  \draw[EdgeStyle, style={stealth-, orange, ultra thick}](v5) to (f5);
  \draw[EdgeStyle, style={-stealth}](v1) to (f2);
  \draw[EdgeStyle, style={stealth-, orange, ultra thick}](v2) to (f3);
  \draw[EdgeStyle, style={stealth-, orange, ultra thick}](v3) to (f4);
  \draw[EdgeStyle, style={-stealth}](v4) to (f5);
  \draw[EdgeStyle, style={stealth-, orange, ultra thick}](v4) to (f2);
  \end{tikzpicture}
  \caption{Segmentation ($M_2$)}
  \label{fig:supplement strongly connected components 2}
\end{subfigure}\hfill
\begin{subfigure}[b]{0.3\linewidth}
  \centering
  \begin{tikzpicture}[scale=0.65,every node/.style={transform shape}]
  \GraphInit[vstyle=Normal]
  \SetGraphUnit{1}
  \Vertex[L=$v_1$, x=0.0, y=0.0] {v1}
  \Vertex[L=$v_2$, x=1.2, y=0.0] {v2}
  \Vertex[L=$v_3$, x=2.2, y=0.0] {v3}
  \Vertex[L=$v_4$, x=3.2, y=0.0] {v4}
  \Vertex[L=$v_5$, x=4.4, y=0.0] {v5}

  \Vertex[L=$f_1$, x=0.0, y=-1.0] {f1}
  \Vertex[L=$f_2$, x=1.2, y=-1.0] {f2}
  \Vertex[L=$f_3$, x=2.2, y=-1.0] {f3}
  \Vertex[L=$f_4$, x=3.2, y=-1.0] {f4}
  \Vertex[L=$f_5$, x=4.4, y=-1.0] {f5}

  \Vertex[L=$w_1$, x=-0.2, y=-2.2, style={dashed}] {w1}
  \Vertex[L=$w_2$, x=1.0, y=-2.2, style={dashed}] {w2}
  \Vertex[L=$w_3$, x=2.2, y=-2.2, style={dashed}] {w3}
  \Vertex[L=$w_4$, x=3.4, y=-2.2, style={dashed}] {w4}
  \Vertex[L=$w_5$, x=4.6, y=-2.2, style={dashed}] {w5}

  \node[draw=black, fit=(v1) (f1), inner sep=0.1cm ]{};
  \node[draw=black, fit=(v2) (f2) (v3) (f3) (v4) (f4), inner sep=0.1cm ]{};
  \node[draw=black, fit=(v5) (f5), inner sep=0.1cm ]{};

  \node[draw=black, fit=(w1), inner sep=0.1cm ]{};
  \node[draw=black, fit=(w2), inner sep=0.1cm ]{};
  \node[draw=black, fit=(w3), inner sep=0.1cm ]{};
  \node[draw=black, fit=(w4), inner sep=0.1cm ]{};
  \node[draw=black, fit=(w5), inner sep=0.1cm ]{};

  \draw[EdgeStyle, style={-stealth, very thick}](v1) to (0.7,0.0);
  \draw[EdgeStyle, style={-stealth, very thick}](v4) to (3.9,0.0);
  \draw[EdgeStyle, style={-stealth, very thick}](w1) to (0.0,-1.5);
  \draw[EdgeStyle, style={-stealth, very thick}](w2) to (1.2,-1.5);
  \draw[EdgeStyle, style={-stealth, very thick}](w3) to (2.2,-1.5);
  \draw[EdgeStyle, style={-stealth, very thick}](w4) to (3.2,-1.5);
  \draw[EdgeStyle, style={-stealth, very thick}](w5) to (4.4,-1.5);
  \end{tikzpicture}
  \caption{Causal ordering graph.}
  \label{fig:supplement causal ordering graph}
\end{subfigure}
\caption{Graphical illustration of the causal ordering algorithm that was described in Section~\ref{sec:introduction:causal ordering graph}. Figure~\ref{fig:supplement bipartite graph} shows the bipartite graph that is associated with equations \eqref{eq:f1} to \eqref{eq:f5}. Application of the first step of the causal ordering algorithm results in the directed graph in Figure~\ref{fig:supplement directed graph 1} for perfect matching $M_1$ and that in Figure~\ref{fig:supplement directed graph 2} for perfect matching $M_2$. The blue and orange edges correspond to the edges in the perfect matchings $M_1$ and $M_2$, respectively. Figures \ref{fig:supplement strongly connected components 1} and \ref{fig:supplement strongly connected components 2} show that the segmentation into strongly connected components does not depend on the choice of the perfect matching. Exogenous vertices and edges from these vertices to clusters were added to the causal ordering graph in Figure~\ref{fig:supplement causal ordering graph}.}
\label{fig:supplement causal ordering algorithm}
\end{figure}
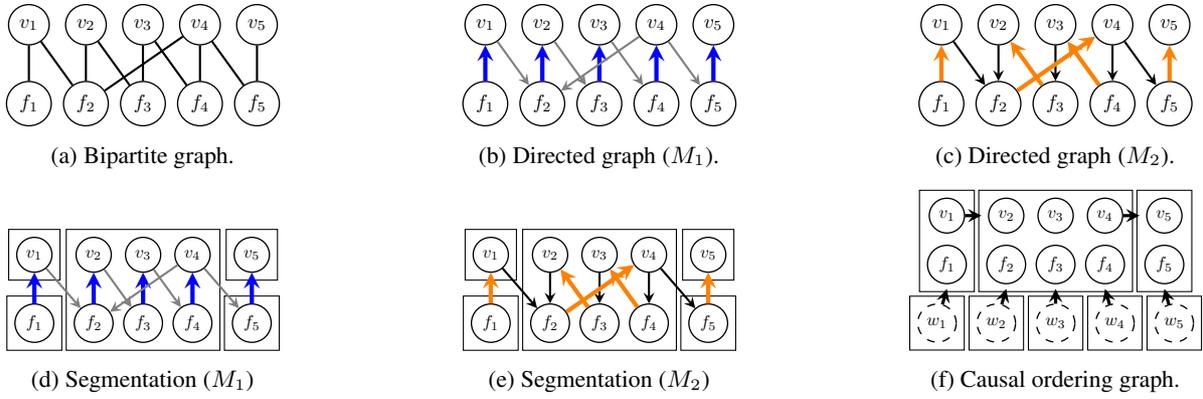

\subsection{Example: signaling cascade model}\label{sec:supplement:signaling_cascade_model}

In this supplementary section we provide more details on the signaling cascade model that is discussed in Section~\ref{sec:signaling_cascade_model}.

We denote the concentrations of active (phosphorylated) RAS, RAF, MEK, and ERK proteins, respectively, by
$X_{s}(t)$, $X_{r}(t)$, $X_{m}(t)$, and $X_{e}(t)$, and denote by $I(t)$ an external stimulus or perturbation.
The system dynamics is modeled by differential equations:
\begin{align}
\label{eq:mapk s}
\dot{X}_{s}(t) &= I(t) \frac{k_{Is} \left(T_s-X_s(t)\right)}{\left(K_{Is} + (T_s-X_s(t)) \right) \left(1+\left(\frac{X_e(t)}{K_e}\right)^{\frac{3}{2}}\right) } -  F_{s} k_{F_s s} \frac{X_s(t)}{K_{F_s s} + X_s(t)} \\
\label{eq:mapk r}
\dot{X}_r(t) &= \frac{X_s(t) k_{sr} (T_r - X_r(t))}{K_{sr} + (T_r - X_r(t))} - F_r k_{F_r r} \frac{X_r(t)}{K_{F_r r} + X_r(t)} \\
\label{eq:mapk m}
\dot{X}_m(t) &= \frac{X_r(t) k_{rm} (T_m - X_m(t))}{K_{rm} + (T_m - X_m(t))} - F_m k_{F_m m} \frac{X_m(t)}{K_{F_m m} + X_m(t)} \\
\label{eq:mapk e}
\dot{X}_e(t) &= \frac{X_m(t) k_{me} (T_e - X_e(t))}{K_{me} + (T_e - X_e(t))} - F_e k_{F_e e} \frac{X_e(t)}{K_{F_e e} + X_e(t)}.
\end{align}
These dynamical equations correspond with a signaling pathway that goes from $I(t)$ to $X_s(t)$ to $X_r(t)$ to $X_m(t)$ to $X_e(t)$ with negative feedback from $X_e(t)$ on $X_s(t)$.
We will study this system in a certain saturated regime;
specifically, for $(T_e-X_e(t))\gg K_{me}$ and $X_e(t)\gg K_{F_e e}$ the following approximation of (\ref{eq:mapk e}) holds:
\begin{align}
\label{eq:mapk e approx}
\dot{X}_e(t) \approx X_m(t) k_{me} - F_e k_{F_e e}.
\end{align}

Thus, the saturated dynamical model that we consider consists of differential equations \eqref{eq:mapk s}, \eqref{eq:mapk r}, \eqref{eq:mapk m} and \eqref{eq:mapk e approx}.
The corresponding equilibrium equations of the saturated model are given by:
\begin{align}
\label{eq:mapk s eq}
0 &= I \frac{k_{Is} \left(T_s-X_s\right)}{\left(K_{Is} + (T_s-X_s) \right) \left(1+\left(\frac{X_e}{K_e}\right)^{\frac{3}{2}}\right) } -  F_{s} k_{F_s s} \frac{X_s}{K_{F_s s} + X_s} \\
\label{eq:mapk r eq}
0 &= \frac{X_s k_{sr} (T_r - X_r)}{K_{sr} + (T_r - X_r)} - F_r k_{F_r r} \frac{X_r}{K_{F_r r} + X_r} \\
\label{eq:mapk m eq}
0 &= \frac{X_r k_{rm} (T_m - X_m)}{K_{rm} + (T_m - X_m)} - F_m k_{F_m m} \frac{X_m}{K_{F_m m} + X_m} \\
\label{eq:mapk e eq approx}
0 &= X_m k_{me} - F_e k_{F_e e},
\end{align}
where we also assume the input signal $I$ to be stationary (constant in time).

Figure~\ref{fig:protein pathway} shows the results of applying the causal ordering procedure to the full model, and to the partial model that treats the equilibrium ERK concentration $X_e$ as unobserved and assumes it to be exogenous with respect to the observed concentrations $X_s$, $X_r$ and $X_m$ of RAS, RAF and MEK, respectively.

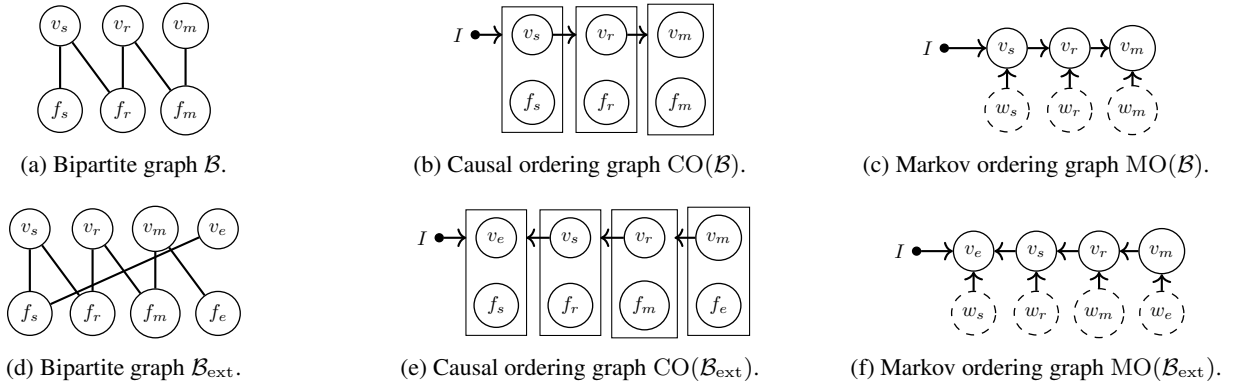
\begin{figure}[ht]
	\begin{subfigure}[b]{0.3\textwidth}
		\centering
		\vspace*{3mm}
		\begin{tikzpicture}[scale=0.75,every node/.style={transform shape}]
		\GraphInit[vstyle=Normal]
		\SetGraphUnit{1}
		\Vertex[L=$v_s$,x=0,y=0] {vs}
		\Vertex[L=$v_r$,x=1.1,y=0] {vr}
		\Vertex[L=$v_m$,x=2.2,y=0] {vm}
		\Vertex[L=$f_s$,x=0,y=-1.5] {fs}
		\Vertex[L=$f_r$,x=1.1,y=-1.5] {fr}
		\Vertex[L=$f_m$,x=2.2,y=-1.5] {fm}
		\draw[EdgeStyle, style={-}](vs) to (fs);
		\draw[EdgeStyle, style={-}](vs) to (fr);
		\draw[EdgeStyle, style={-}](vr) to (fr);
		\draw[EdgeStyle, style={-}](vr) to (fm);
		\draw[EdgeStyle, style={-}](vm) to (fm);
		\end{tikzpicture}
    \caption{Bipartite graph $\BG$.}
		\label{fig:protein pathway:submodel:bipartite graph}
	\end{subfigure}%
	\begin{subfigure}[b]{0.4\textwidth}
		\centering
		\begin{tikzpicture}[scale=0.75,every node/.style={transform shape}]
		\GraphInit[vstyle=Normal]
		\SetGraphUnit{1}
		\Vertex[L=$v_m$,x=2.6,y=0] {vm}
		\Vertex[L=$v_r$,x=1.3,y=0] {vr}
		\Vertex[L=$v_s$,x=0.0,y=0] {vs}
		\Vertex[L=$f_m$,x=2.6,y=-1.2] {fm}
		\Vertex[L=$f_r$,x=1.3,y=-1.2] {fr}
		\Vertex[L=$f_s$,x=0.0,y=-1.2] {fs}
		
		\begin{scope}[VertexStyle/.append style = {minimum size = 4pt, 
			inner sep = 0pt,
			color=black}]
		\Vertex[x=-1.0, y=0.0, L=$I$, Lpos=180, LabelOut]{I}
		\end{scope}

    \node[draw=black, fit=(vm) (fm), inner sep=0.1cm]{};
    \node[draw=black, fit=(vr) (fr), inner sep=0.1cm]{};
		\node[draw=black, fit=(vs) (fs), inner sep=0.1cm]{};
		
		\draw[EdgeStyle, style={->}](I) to (-0.525,0.0);
		\draw[EdgeStyle, style={->}](vs) to (0.775,0.0);
		\draw[EdgeStyle, style={->}](vr) to (2.025,0.0);
		\end{tikzpicture}
    \caption{Causal ordering graph $\mathrm{CO}(\BG)$.}
		\label{fig:protein pathway:submodel:causal ordering graph}
	\end{subfigure}%
	\begin{subfigure}[b]{0.3\textwidth}
		\centering
		\begin{tikzpicture}[scale=0.75,every node/.style={transform shape}]
		\GraphInit[vstyle=Normal]
		\SetGraphUnit{1}
		\Vertex[L=$v_s$, x=1.1, y=0] {vS}
		\Vertex[L=$v_r$, x=2.2, y=0] {vR}
		\Vertex[L=$v_m$, x=3.3, y=0] {vM}
		\Vertex[L=$w_{s}$, x=1.1, y=-1.1, style={dashed}] {wS}
		\Vertex[L=$w_{r}$, x=2.2, y=-1.1, style={dashed}] {wR}
 		\Vertex[L=$w_{m}$, x=3.3, y=-1.1, style={dashed}] {wM}
		
		\begin{scope}[VertexStyle/.append style = {minimum size = 4pt, 
			inner sep = 0pt,
			color=black}]
		\Vertex[x=0.0, y=0.0, L=$I$, Lpos=180, LabelOut]{I}
		\end{scope}
    \draw[EdgeStyle, style={->}](I) to (vS);
    \draw[EdgeStyle, style={->}](wS) to (vS);
    \draw[EdgeStyle, style={->}](wR) to (vR);
    \draw[EdgeStyle, style={->}](wM) to (vM);
    \draw[EdgeStyle, style={->}](vR) to (vM);
    \draw[EdgeStyle, style={->}](vS) to (vR);
		\end{tikzpicture}
		\caption{Markov ordering graph $\mathrm{MO}(\BG)$.}
		\label{fig:protein pathway:submodel:markov ordering graph2}
	\end{subfigure}
	\begin{subfigure}[b]{0.3\textwidth}
		\centering
		\vspace*{3mm}
		\begin{tikzpicture}[scale=0.75,every node/.style={transform shape}]
		\GraphInit[vstyle=Normal]
		\SetGraphUnit{1}
		\Vertex[L=$v_s$,x=0,y=0] {vs}
		\Vertex[L=$v_r$,x=1.1,y=0] {vr}
		\Vertex[L=$v_m$,x=2.2,y=0] {vm}
		\Vertex[L=$v_e$,x=3.3,y=0] {ve}
		\Vertex[L=$f_s$,x=0,y=-1.5] {fs}
		\Vertex[L=$f_r$,x=1.1,y=-1.5] {fr}
		\Vertex[L=$f_m$,x=2.2,y=-1.5] {fm}
		\Vertex[L=$f_e$,x=3.3,y=-1.5] {fe}
		\draw[EdgeStyle, style={-}](vs) to (fs);
		\draw[EdgeStyle, style={-}](vs) to (fr);
		\draw[EdgeStyle, style={-}](vr) to (fr);
		\draw[EdgeStyle, style={-}](vr) to (fm);
		\draw[EdgeStyle, style={-}](vm) to (fm);
		\draw[EdgeStyle, style={-}](vm) to (fe);
		\draw[EdgeStyle, style={-}](ve) to (fs);
		\end{tikzpicture}
    \caption{Bipartite graph $\BG_{\mathrm{ext}}$.}
		\label{fig:protein pathway:extended model:bipartite graph}
	\end{subfigure}%
	\begin{subfigure}[b]{0.4\textwidth}
		\centering
		\begin{tikzpicture}[scale=0.75,every node/.style={transform shape}]
		\GraphInit[vstyle=Normal]
		\SetGraphUnit{1}
		\Vertex[L=$v_m$,x=3.9,y=0] {vm}
		\Vertex[L=$v_r$,x=2.6,y=0] {vr}
		\Vertex[L=$v_s$,x=1.3,y=0] {vs}
		\Vertex[L=$v_e$,x=0.0,y=0] {ve}
		\Vertex[L=$f_e$,x=3.9,y=-1.2] {fe}
		\Vertex[L=$f_m$,x=2.6,y=-1.2] {fm}
		\Vertex[L=$f_r$,x=1.3,y=-1.2] {fr}
		\Vertex[L=$f_s$,x=0.0,y=-1.2] {fs}
		
		\begin{scope}[VertexStyle/.append style = {minimum size = 4pt, 
			inner sep = 0pt,
			color=black}]
		\Vertex[x=-1.0, y=0.0, L=$I$, Lpos=180, LabelOut]{I}
		\end{scope}

    \node[draw=black, fit=(vm) (fe), inner sep=0.1cm]{};
		\node[draw=black, fit=(vr) (fm), inner sep=0.1cm]{};
		\node[draw=black, fit=(vs) (fr), inner sep=0.1cm]{};
		\node[draw=black, fit=(ve) (fs), inner sep=0.1cm]{};
		
		\draw[EdgeStyle, style={->}](I) to (-0.525,0.0);
		\draw[EdgeStyle, style={->}](vs) to (0.525,0.0);
		\draw[EdgeStyle, style={->}](vr) to (1.825,0.0);
		\draw[EdgeStyle, style={->}](vm) to (3.125,0.0);
		\end{tikzpicture}
		\caption{Causal ordering graph $\mathrm{CO}(\BG_{\mathrm{ext}})$.}
		\label{fig:protein pathway:extended model:causal ordering graph}
	\end{subfigure}%
	\begin{subfigure}[b]{0.3\textwidth}
		\centering
		\begin{tikzpicture}[scale=0.75,every node/.style={transform shape}]
		\GraphInit[vstyle=Normal]
		\SetGraphUnit{1}
		\Vertex[L=$v_e$, x=0.0, y=0] {vE}
		\Vertex[L=$v_s$, x=1.1, y=0] {vS}
		\Vertex[L=$v_r$, x=2.2, y=0] {vR}
		\Vertex[L=$v_m$, x=3.3, y=0] {vM}
		\Vertex[L=$w_{s}$, x=0, y=-1.1, style={dashed}] {wS}
		\Vertex[L=$w_{r}$, x=1.1, y=-1.1, style={dashed}] {wR}
		\Vertex[L=$w_{m}$, x=2.2, y=-1.1, style={dashed}] {wM}
		\Vertex[L=$w_{e}$, x=3.3, y=-1.1, style={dashed}] {wE}
		
		\begin{scope}[VertexStyle/.append style = {minimum size = 4pt, 
			inner sep = 0pt,
			color=black}]
		\Vertex[x=-1.0, y=0.0, L=$I$, Lpos=180, LabelOut]{I}
		\end{scope}
		\draw[EdgeStyle, style={->}](I) to (vE);
		\draw[EdgeStyle, style={->}](wS) to (vE);
		\draw[EdgeStyle, style={->}](wR) to (vS);
		\draw[EdgeStyle, style={->}](wM) to (vR);
		\draw[EdgeStyle, style={->}](wE) to (vM);
		\draw[EdgeStyle, style={->}](vM) to (vR);
		\draw[EdgeStyle, style={->}](vR) to (vS);
		\draw[EdgeStyle, style={->}](vS) to (vE);
		\end{tikzpicture}
    \caption{Markov ordering graph $\mathrm{MO}(\BG_{\mathrm{ext}})$.}
		\label{fig:protein pathway:extended model:markov ordering graph2}
	\end{subfigure}
  \caption{Graphs associated with the saturated protein signaling pathway model, where indices $s,r,m,e$ correspond to concentrations of active RAS, RAF, MEK and ERK respectively, and $I$ is an exogenous input signal.
  Top row: submodel for RAS, RAF and MEK only.
  Bottom row: model extension with ERK.}
	\label{fig:protein pathway}
\end{figure}

\subsection{Proofs}
\label{sec:supplement:proofs}

\begin{manualtheorem}{\ref{thm:preserved causal predictions}}
  Consider model equations $F$ containing endogenous variables $V$ with bipartite graph $\BG$. Suppose $F$ is extended with equations $F_+$ containing endogenous variables in $V\cup V_+$, where $V_+$ contains endogenous variables that are added by the model extension (which may include parameters or exogenous variables that appear in $F$ and become endogenous in the extended model). Let $\BG_{\mathrm{ext}}$ be the bipartite graph associated with $F_{\mathrm{ext}}=F\cup F_+$ and $V_{\mathrm{ext}}=V\cup V_+$, and $\BG_+$ the bipartite graph associated with the extension $F_+$ and $V_+$, where variables in $V$ appearing in $F^+$ are treated as exogenous variables (i.e.\ they are not added as vertices in $\BG_+$). If $\BG$ and $\BG_+$ both have a perfect matching then:
  \begin{enumerate}
    \item $\BG_{\mathrm{ext}}$ has a perfect matching,
    \item ancestral relations in $\mathrm{CO}(\BG)$ are also present in $\mathrm{CO}(\BG_{\mathrm{ext}})$,
    \item d-connections in $\mathrm{MO}(\BG)$ are also present in $\mathrm{MO}(\BG_{\mathrm{ext}})$.
  \end{enumerate}
\end{manualtheorem}
\begin{proof}
The causal ordering graph $\mathrm{CO}(\BG)$ is constructed from a perfect matching $M$ for the bipartite graph $\BG=\tuple{V,F,E}$. Let $M_+$ be a perfect matching for $\BG_+$. Note that $M_{\mathrm{ext}}=M\cup M_+$ is a perfect matching for $\BG_{\mathrm{ext}}=\tuple{V\cup V_+, F\cup F_+, E_{\mathrm{ext}}}$. Following the causal ordering algorithm for $\BG, M$ and $\BG_{\mathrm{ext}}, M_{\mathrm{ext}}$, we note that $\G(\BG, M)$ is a subgraph of $\G(\BG_{\mathrm{ext}},M_{\mathrm{ext}})$ and hence clusters in $\mathrm{CO}(\BG)$ are fully contained in clusters in $\mathrm{CO}(\BG_{\mathrm{ext}})$. Therefore ancestral relations in $\mathrm{CO}(\BG)$ are also present in $\mathrm{CO}(\BG_{\mathrm{ext}})$.

It follows directly from the definition \citep{Forre2017} that $\sigma$-connections in a graph remain present if the graph is extended with additional vertices and edges. The directed graphs $\G(\BG, M)$ and $\G(\BG_{\mathrm{ext}},M_{\mathrm{ext}})$ can be augmented with exogenous variables by adding exogenous vertices to these graphs with directed edges towards the equations in which they appear. The $\sigma$-connections in the augmentation of $\G(\BG, M)$ must also be present in the augmentation of $\G(\BG_{\mathrm{ext}},M_{\mathrm{ext}})$. By \citep[Corollary 2.8.4, ][]{Forre2017} and \citep[Lemma~43, ][]{Blom2020} we have that d-connections in $\mathrm{MO}(\BG)$ must also be present in $\mathrm{MO}(\BG_{\mathrm{ext}})$.
\end{proof}

\begin{manualtheorem}{\ref{thm:preserveration of conditional independence relations}}
  Let $F$, $F_+$, $F_{\mathrm{ext}}$, $V$, $V_+$, $V_{\mathrm{ext}}$, $\BG$, $\BG_+$, and $\BG_{\mathrm{ext}}$ be as in Theorem~\ref{thm:preserved causal predictions}. If $\BG$ and $\BG_+$ both have perfect matchings and no vertex in $V_+$ is adjacent to a vertex in $F$ in $\BG_{\mathrm{ext}}$ then:
  \begin{enumerate}
    \item ancestral relations absent in $\mathrm{CO}(\BG)$ are also absent in $\mathrm{CO}(\BG_{\mathrm{ext}})$,
    \item d-connections absent in $\mathrm{MO}(\BG)$ are also absent in $\mathrm{MO}(\BG_{\mathrm{ext}})$.
  \end{enumerate}
\end{manualtheorem}
\begin{proof}
Since $\BG$ and $\BG_+$ both have perfect matchings the results of Theorem~\ref{thm:preserved causal predictions} hold. Let $\G(\BG,M)$, and $\G(\BG_{\mathrm{ext}}, M_{\mathrm{ext}})$ be as in the proof of Theorem~\ref{thm:preserved causal predictions}. Note that in $M_{\mathrm{ext}}$ vertices in $F_+$ are matched to vertices in $V_+$ and therefore edges between $f_+ \in F_+$ and $v\in\adj{\BG_{\mathrm{ext}}}{F_+}\setminus V_+$ are oriented as $(f_+\leftarrow v)$ in $\G(\BG_{\mathrm{ext}}, M_{\mathrm{ext}})$. By assumption, we therefore have that vertices in $V_+$ are non-ancestors of vertices in $V\cup F$ in $\G(\BG_{\mathrm{ext}}, M_{\mathrm{ext}})$. Since $M\subseteq M_{\mathrm{ext}}$ we know that the same directed edges between vertices in $V$ and $F$ appear in both $\G(\BG,M)$ and $\G(\BG_{\mathrm{ext}}, M_{\mathrm{ext}})$. Notice that the subgraph of $\G(\BG_{\mathrm{ext}}, M_{\mathrm{ext}})$ induced by the vertices $V\cup F$ coincides with $\G(\BG,M)$. Hence $\mathrm{CO}(\BG)$ is the induced subgraph of $\mathrm{CO}(\BG_{\mathrm{ext}})$ and $\mathrm{MO}(\BG)$ is the induced subgraph of $\mathrm{MO}(\BG_{\mathrm{ext}})$.
\end{proof}

\begin{lemma}
\label{lemma:selfregulating perfect matching}
Consider a first-order dynamical model in canonical form for endogenous variables $V$ and let $F$ be the equilibrium equations of the model. If all variables in $V$ are self-regulating then $\BG$ has a perfect matching.
\end{lemma}
\begin{proof}
Recall that the equilibrium equation constructed from the derivative of a variable $i$ is labelled $f_i$ according to the natural labelling. When a variable in $v_i\in V$ is self-regulating then it can be matched to its equilibrium equation $f_i$. If this holds for all variables in $V$ then $\BG$ has a perfect matching.
\end{proof}

\begin{lemma}
\label{lemma:uniqueness of directed graph}
Let $\BG$ be a bipartite graph and let $M$ and $M'$ be two distinct perfect matchings. The associated directed graphs $\G(\BG,M)$ and $\G(\BG,M')$ that are obtained in step \ref{alg:step 1} of the causal ordering algorithm differ only in the direction of cycles.
\end{lemma}
\begin{proof}
  This follows directly from the fact that the output of the causal ordering
  algorithm does not depend on the choice of the perfect matching. This result
  is a direct consequence of Theorem 4 and Theorem 6 in \citet{Blom2020}.
\end{proof}

\begin{manualtheorem}{\ref{thm:feedback}}
Consider a first-order dynamical model in canonical form for endogenous variables $V$ and an extension consisting of canonical first-order differential equations for additional endogenous variables $V_+$. Let $F$ and $F_{\mathrm{ext}}= F\cup F_+$ be the equilibrium equations of the original and extended model respectively. Let $\BG=\tuple{V,F,E}$ be the bipartite graph associated with $F$ and $\BG_{\mathrm{ext}}=\tuple{V_{\mathrm{ext}}, F_{\mathrm{ext}}, E_{\mathrm{ext}}}$ the bipartite graph associated with $F_{\mathrm{ext}}$. Assume that $\BG$ and $\BG_{\mathrm{ext}}$ both have perfect matchings. If the model extension does not introduce a new feedback loop with the original dynamical model, then d-connections in $\mathrm{MO}(\BG)$ are also present in $\mathrm{MO}(\BG_{\mathrm{ext}})$.
\end{manualtheorem}
\begin{proof}
Let $E_{\mathrm{nat}}$ be the set of edges $(v_i-f_i)$ associated with the natural labelling of the equilibrium equations of the extended dynamical model. Note that the feedback loops in the dynamical model coincide with cycles in the directed graph $\G(\BG_{\mathrm{nat}}, M_{\mathrm{nat}})$ that is obtained by applying step \ref{alg:step 1} of the causal ordering algorithm to the bipartite graph $\BG_{\mathrm{nat}}=\tuple{V_{\mathrm{ext}},F_{\mathrm{ext}},E_{\mathrm{ext}}\cup E_{\mathrm{nat}}}$ using the perfect matching $M_{\mathrm{nat}} = E_{\mathrm{nat}}$.

By Theorem~\ref{thm:preserved causal predictions}, we know that if $\BG$ and $\BG_+$ (the subgraph of $\BG_{\mathrm{ext}}$ induced by $V_+\cup F_+$) both have perfect matchings then d-connections in $\mathrm{MO}(\BG)$ must also be present in $\mathrm{MO}(\BG_{\mathrm{ext}})$. Therefore, if there exists a perfect matching $M_{\mathrm{ext}}$ for $\BG_{\mathrm{ext}}$ so that each $f\in F$ is $M_{\mathrm{ext}}$-matched to a vertex $v\in V$ and each $f_+\in F_+$ is $M_{\mathrm{ext}}$-matched to a vertex $v_+\in V_+$ in $\BG_{\mathrm{ext}}$, d-connections in $\mathrm{MO}(\BG)$ are also present in $\mathrm{MO}(\BG_{\mathrm{ext}})$.

We will prove the contrapositive of the theorem, so we start with the assumption that the d-connections in $\mathrm{MO}(\BG)$ are not preserved in $\mathrm{MO}(\BG_{\mathrm{ext}})$. In that case, there must exist a perfect matching $M_{\mathrm{ext}}$ for $\BG_{\mathrm{ext}}$ so that there is an $f\in F$ that is $M_{\mathrm{ext}}$-matched to a $v_+\in V_+$ and a $v\in V$ that is $M_{\mathrm{ext}}$-matched to a $f_+\in F_+$. Note that since $\BG_{\mathrm{ext}}$ is a subgraph of $\BG_{\mathrm{nat}}$, this perfect matching $M_{\mathrm{ext}}$ is also a perfect matching for $\BG_{\mathrm{nat}}$. Lemma~\ref{lemma:uniqueness of directed graph} says that $\G(\BG_{\mathrm{nat}}, M_{\mathrm{nat}})$ and $\G(\BG_{\mathrm{nat}}, M_{\mathrm{ext}})$ only differ in the direction of cycles. We know that vertices in $V$ are only $M_{\mathrm{nat}}$-matched to vertices in $F$, while vertices in $V_+$ are only $M_{\mathrm{nat}}$-matched to vertices in $F_+$. Therefore, the vertices $v_+$ and $f$ must be on a directed cycle in both directed graphs, as well as $v$ and $f_+$. Hence the model extension $F_+$ introduced a new feedback loop that includes variables in the original model.
\end{proof}

\end{document}